%% file: main.tex
\documentclass[letterpaper,twocolumn,10pt]{article}
\usepackage{usenix}
\usepackage{amssymb}
\usepackage{amsmath,breqn,graphicx,wrapfig,paralist,eurosym,amsthm}
\usepackage{enumitem}

\usepackage[T1]{fontenc}
\usepackage{hyperref}
\usepackage{url}

\ifnum\pdfstrcmp{\jobname}{main}=0
\def\anonymize{0}
\def\release{1}
\def\shrink{1}
\def\diagram{system_diagram.pdf}
\fi

\ifnum\pdfstrcmp{\jobname}{anonymized}=0
\def\anonymize{1}
\def\release{1}
\def\shrink{1}
\def\diagram{system_diagram_anonymized.pdf}
\fi

\ifnum\pdfstrcmp{\jobname}{release}=0
\def\anonymize{0}
\def\release{1}
\def\shrink{0}
\def\diagram{system_diagram.pdf}
\fi

\if\shrink1
\usepackage[subtle]{savetrees}
\fi

\newcommand{\Tat}{T{\^a}tonnement}
\newcommand\asset[1]{\mathcal{#1}}
\newcommand{\assetset}{\mathfrak{A}}
\newcommand{\assetsize}{\vert \mathfrak{A}\vert}

\if\release1
\newcommand\XXX[1]{}
\else
\newcommand\XXX[1]{\begingroup \bfseries\color{red} #1\endgroup}
\fi

\input{name}

\newtheorem{theorem}{Theorem}
\newtheorem{corollary}{Corollary}
\newtheorem{definition}{Definition}
\newtheorem{example}{Example}

\DeclareMathOperator*{\argmax}{arg\,max}

\if\shrink1
\usepackage[compact]{titlesec}
\fi

\usepackage{caption}

\usepackage[title]{appendix}

\begin{document}

\if\shrink1
\captionsetup[figure]{labelfont={bf},name={Fig.},labelsep=period,textfont=it,belowskip=-15pt,skip=-1pt}
\else
\captionsetup[figure]{labelfont={bf},name={Fig.},labelsep=period,textfont=it}
\fi

\date{} 

\title{\Large \bf \SPEEDEX{}:\\
A Scalable, Parallelizable, and Economically Efficient Decentralized EXchange
}


\if\anonymize1
\author{Submission 224}
\else
\input{authors}
\fi


\maketitle

\begin{abstract}
\vspace{5pt}
\input{abstract}
\end{abstract}

\input{intro}

\input{architecture}
\input{tatonnement}

\input{tatonnement_evaluation}

\input{scalability_evaluation}

\input{limitations}

\input{system_design}
\input{related_work}

\input{conclusion}

\if\anonymize1

\else
\input{acks}

\fi

\bibliographystyle{plain}
\bibliography{osdi.bib}
\newpage
~
\newpage
\begin{appendices}
\input{arrow_debreu_mapping}

\input{approximation}
\input{tatonnement_mods}
\input{linear_program}

\input{decomposition}

\input{alternate_strategies}

\input{tatonnement_preprocessing}

\input{ppad}

\input{filtering_performance}

\input{bstm_compare}

\input{extra_design}

\input{morereplicas}

\end{appendices}
\end{document}

%% file: name.tex
\if\anonymize1

\newcommand{\SPEEDEX}{GeckoDEX}
\newcommand{\stellar}{Blockchain X}
\newcommand{\thestellarbc}{Blockchain X}
\newcommand{\Thestellarbc}{Blockchain X}

\newcommand{\speedexgithub}{<redacted>}

\else

\newcommand{\SPEEDEX}{SPEEDEX\@}

\newcommand{\stellar}{Stellar}
\newcommand{\thestellarbc}{the Stellar blockchain}
\newcommand{\Thestellarbc}{The Stellar blockchain}

\newcommand{\speedexgithub}{\url{https://github.com/scslab/speedex}}
\fi

%% file: authors.tex
\author{
{\rm Geoffrey Ramseyer}\\
Stanford University
\and
{\rm Ashish Goel}\\
Stanford University
\and
{\rm David Mazières}\\
Stanford University
} 

%% file: abstract.tex
\SPEEDEX{}
is a decentralized exchange (DEX) that lets participants
securely trade assets without giving any single party undue control
over the market.  \SPEEDEX{} offers several advantages over prior DEXes.
It achieves high throughput---over 200,000 transactions per second on
48-core servers, even with tens of millions of open offers.  
\SPEEDEX{} runs entirely within a Layer-1 blockchain, and thus achieves 
its scalability without fragmenting market liquidity between multiple blockchains or rollups.  
It eliminates internal
arbitrage opportunities, so that a direct trade from asset $\asset{A}$ to
asset $\asset{B}$
always receives as good a price as trading through some third asset
such as USD\@.  Finally, it prevents certain front-running attacks that would
otherwise increase the effective bid-ask spread for small traders.
\SPEEDEX{}'s key design insight is its use of an Arrow-Debreu exchange market
structure that fixes the valuation of assets for all trades in a given
block of transactions.  
We construct an algorithm, which is both asymptotically efficient and empirically practical,
that computes these 
valuations while exactly preserving a DEX's financial correctness constraints.
Not only does this market structure provide
fairness across trades, but it also makes trade operations commutative and
hence efficiently parallelizable.  \SPEEDEX{} is prototyped but not yet merged
within the Stellar blockchain, one of the largest Layer-1 blockchains.
\XXX{This needs to be specific for the camera-ready}

%% file: intro.tex
\section{Introduction}

Digital currencies are moving closer to mainstream adoption.  Examples
include central bank digital currencies (CBDCs) such as China's DC/EP \cite{dcep},
commercial efforts \cite{helvetia,stellarcbdc}, and many
decentralized-blockchain-based stablecoins such as Tether \cite{tether}, Dai \cite{dai}, and
USDC \cite{usdc}\@.  These currencies vary wildly in terms of
privacy, openness, smart contract support, performance, regulatory
risk, solvency guarantees, compliance features, retail vs.\ wholesale
suitability, and centralization of the underlying ledger.  Because of
these differences, and because financial stability demands different
monetary policy in different countries, we cannot hope for a
one-size-fits-all global digital currency.  Instead, to realize the
full potential of digital currencies (and digital assets in general), we need an ecosystem
where many digital currencies can efficiently interoperate.

Effective interoperability requires an {\it exchange}:
 an efficient system for exchanging
one digital asset for another.  Users post offers to trade one asset for another on
the exchange, and then the exchange matches mutually compatible offers together
and transfers assets according to the offered terms.  For example, one user might
offer to trade 110 USD for 100 EUR, and might be matched against another user
who previously offered to trade 100 EUR for 110 USD.  A typical
exchange maintains \emph{orderbooks}
of all of the open trade offers.

The ideal digital currency exchange should, at minimum,
\begin{compactitem}
    \item not give any central authority undue power over the global
    flow of money,
    \item operate transparently and auditably,
    \item give every user an equal level of access,
    \item  enable
    efficient trading between every pair of currencies
    (make effective use of all available liquidity), and
    \item support arbitrarily high throughput,
    without charging significant fees to users.
\end{compactitem}

Scalability is crucial for a piece of financial infrastructure that must last
far into the future, as the number of individuals transacting internationally continues
to grow.  Furthermore, the above feature list is by no means complete; a deployment may want any
number of additional features, such as persistent logging, simplified payment verification \cite{nakamoto:bitcoin},
or integrations with legacy systems, each of which slows down the system's performance.
Scalability, viewed from another angle, enables the system to add features without decreasing
overall transaction throughput (at the cost of additional compute hardware).

The gold standard for avoiding centralized
control is a \emph{decentralized exchange}, or DEX:  a
transparent exchange implemented as a deterministic replicated state
machine maintained by many different parties.  To prevent theft, a DEX
requires all transactions to be digitally signed by the relevant asset
holders.  To prevent cheating, replicas organize history into an
append-only blockchain.  Replicas agree on blockchain state
through a Byzantine-fault tolerant consensus protocol, typically some
variant of asynchronous or eventually synchronous Byzantine agreement~\cite{castro:bfs} for
private blockchains or synchronous mining~\cite{nakamoto:bitcoin} for
public ones.

Unfortunately, existing DEX designs cannot meet the last three desiderata.

\paragraph{Equality of Access} In existing exchange designs, users with low-latency connections to an
exchange server (centralized or not) can spy on trades incoming from other users
and \emph{front-run} these trades.  For example, a front-runner might spy an incoming sell offer,
and in response, send a trade that buys and immediately resells an asset at a higher price 
\cite{hackernoonfrontrunning,bancorfrontrunhack}.
In a blockchain, where a block of trades is either finalized entirely or not at all,
this front-running can be made risk-free.  More generally,
some users form special connections with blockchain operators
to gain preferential treatment for their transactions \cite{daian2019flash}.  
This special treatment typically takes the form of ordering transactions in a block
in a favorable manner.  The result is hundreds of millions of dollars siphoned away from 
users \cite{qin2021quantifying}.

\paragraph{Effective Use of Liquidity}
Existing exchange designs are filled with arbitrage opportunities.
A user trading from one currency ~$A$ to another ~$B$ might 
receive a better overall exchange rate by trading through an
intermediate \emph{reserve}
currency ~$C$, such as USD\@.  
Users must typically choose a single (sequence of) intermediate asset(s),
leaving behind arbitrage opportunities with other intermediate assets.
This challenge is especially problematic in the blockchain space, where
market liquidity is typically fragmented between multiple fiat-pegged tokens.

\paragraph{Computational Scalability}

DEX infrastructure
must also be scalable.  The ideal DEX needs to handle as many transactions
per second as users around the globe want to send, without limiting
transaction rates through high fees.  Trading activity growth may outpace Moore's law,
and should not be limited by the rate of increase of
single-CPU-core performance.
An ideal DEX should handle higher transaction rates
simply by using more compute hardware.

Unfortunately, folk wisdom holds that DEXes cannot scale beyond a few thousand transactions per second.
Na\"ive parallel execution would not be replicable across different blockchain nodes.
This wisdom has led to many alternative blockchain scaling techniques,
such as off-chain trade matching~\cite{warren20170x}, automated
market-makers~\cite{uniswapv2}, transaction rollup
systems~\cite{starkware,loopring}, and sharded blockchains\cite{ethsharding} or side-chains \cite{poon2017plasma}.
These
approaches either trust a third party to ensure that orders are
matched with the best available price, or sacrifice the ability to set
traditional limit orders that only sell at or above a certain price
(reducing market liquidity).  Offchain rollup systems, sharded chains, and side-chains further
fragment market liquidity, leading to cross-shard arbitrage and worse exchange rates for
traders.

A challenge for on-chain limit-order DEXes is that the order of
operations affects their results.  Typically, a DEX matches each offer
to the reciprocal offer with the best price: e.g., the first offer
to buy $1$~EUR might consume the only offer priced at $1.09$~USD, leaving
the second to pay $1.10$~USD\@.  
Each trade is a read-modify-write operation on a shared orderbook data structure,
so trades must be serialized.  This serialization order must be
deterministic in a replicated state machine, but na\"ive parallel execution
would make the order of transactions dependent on non-deterministic
thread scheduling.



\subsection{\SPEEDEX{}: Towards an Ideal DEX}

This paper disproves the conventional wisdom about on-chain DEX
performance.  We present \SPEEDEX{}, a fully on-chain decentralized
exchange that meets all of the desiderata
outlined above. \SPEEDEX{} gives every user
an equal level of access (thereby eliminating a widespread class of risk-free
front-running), eliminates internal arbitrage opportunities
(thereby making optimal use of liquidity available on the DEX),
and is
capable of processing over
200,000 transactions per second when deployed on 48-core machines (Figure
\ref{fig:e2e}). \SPEEDEX{} is designed to scale further when given more
hardware.  

Like most blockchains, \SPEEDEX{} processes transactions
in blocks---in our case, a block of 500,000 transactions every few seconds.
Its fundamental principle is that transactions in a block commute:
a block's result is identical regardless of transaction ordering, 
which enables efficient
parallelization~\cite{clements2015scalable}.

\SPEEDEX{}'s core innovation is to execute every order at the same
exchange rate
as every other order in the same block.  
\SPEEDEX{} processes a block of limit orders as one
unified batch, in which, for example, every $1$~EUR sold to buy USD receives exactly
$1.10$~USD in payment.  Furthermore, \SPEEDEX's exchange rates present
no arbitrage opportunities within the exchange; that is, the exchange rate
for trading USD to EUR directly is exactly the exchange rate for
USD to YEN multiplied by the rate for YEN to EUR\@.  These exchange rates
are unique for any (nonempty) batch of trades.
Users interact with \SPEEDEX{}
via traditional limit orders, and \SPEEDEX{} executes a limit order if and only if
the batch's exchange rate exceeds the order's limit price.

This design provides two additional economic advantages.  First, the exchange
offers liquid trading between every asset pair.  Users can directly trade any asset
for any other asset, and the market between these assets will be at least as liquid as the most liquid
market path through intermediate reserve currencies.
Second, \SPEEDEX{} eliminates a class of front-running that is widespread in modern DEXes.
No exchange operator or user with a low-latency network connection can buy an asset
and resell it at a higher price, within the same block.  (Note that this is not
every type of front-running; \S\ref{sec:limitations} and \S\ref{sec:relwork} contrast \SPEEDEX{}'s guarantees
with those of other mitigations, and how they can be combined.)


Furthermore, this economic design enables a scalable systems design that is not
possible using traditional order-matching semantics.  Unlike every other DEX, the operation of \SPEEDEX{}
is efficiently parallelized, allowing \SPEEDEX{} to scale to transaction rates far beyond those seen today.
Transactions within a
block commute with each other precisely because trades all happen at the same
shared set of exchange rates. This means that the transaction processing engine has no need for the sequential
read-modify-update loop of traditional orderbook matching engines.  
Account balances are adjusted using only hardware-level atomics, rather than locking.

\subsection{\SPEEDEX{} Overview}

\SPEEDEX{} is not a blockchain itself; rather, it is a DEX component that can
be integrated into any blockchain.  A copy of the \SPEEDEX{} module
should run inside every replica of a blockchain using the system.
\SPEEDEX{} does not depend on any specific property
of a consensus protocol, but automatically benefits from throughput advances in 
consensus and transaction dissemination (such as \cite{danezis2022narwhal}).
\SPEEDEX{} heavily uses concurrency and benefits from uninterrupted
access to CPU caches, and as such is best implemented directly within blockchain node software (instead of as a smart contract).

We implemented \SPEEDEX{} within a custom blockchain using the
HotStuff consensus
protocol \cite{yin:hotstuff}; this implementation provides the measurements in this paper.
We created a second implementation as a component of
\thestellarbc\cite{lokhava:stellar},
which is
considering a Layer-1 \SPEEDEX{} deployment.

Implementing \SPEEDEX{} introduces both theoretical algorithmic
challenges and systems design challenges.  The core algorithmic challenge is
the computation of the batch prices.
This problem maps to a well-studied problem in the theoretical literature
(equilibrium computation of {\it Arrow-Debreu Exchange Markets}, \S \ref{sec:arrowdebreu});
however, the algorithms in the theoretical literature scale extremely poorly, both
asymptotically and empirically, as the number of open limit orders increases.

We show that the market instances which arise in \SPEEDEX{}
have additional structure not discussed in the theoretical literature, 
and use this structure to build a novel algorithm
(based on the \Tat{} process of \cite{codenotti2005market})
that, in practice, efficiently approximates batch clearing prices.
We then explicitly correct approximation error with a follow-up linear
program.

Our algorithm's runtime is largely independent of the
number of limit orders%
---each \Tat{} query has a runtime of $O(\#$assets$^2\cdot\lg(\#$offers$))$ and the linear program
has size $O(\#$assets$^2)$.  This gives a crucial algorithmic speedup because in the real world,
the number of currencies is much smaller than the number of market
participants. (The experiments of
\S \ref{sec:tateval} and \S\ref{sec:scalable} use 50 assets and tens of
millions of open offers.)

On the systems design side, to implement this exchange, we
design natural commutative transaction semantics
and implement data structures designed for concurrent, batched manipulation and for
efficiently answering queries about the exchange state from the
price computation algorithm.

In recent years, the economics literature has begun discussing
the use of batched trading systems
in traditional markets to combat front-running and externalities associated with high-frequency trading
\cite{aquilina2020quantifying,budish2015high,budishpubliccomment}.  This literature focuses
only on the case of trading between two assets (where price computation is simple)
or where all trades use a single \emph{numeraire} currency~\cite{budish2022flow}.
Our contribution to this line of work is to demonstrate the feasibility of
a batch trading system that exchanges many assets and many numeraire currencies simultaneously,
thereby expanding the design space of implementable market structures.

%% file: architecture.tex
\section{System Architecture}
\label{sec:architecture}

\SPEEDEX{} is an asset exchange implemented as a replicated state machine
in a blockchain architecture (Fig. \ref{fig:diagram}).  
Assets are issued and traded by
\emph{accounts}.  Accounts have public signature keys authorized to
spend their assets.  Signed transactions are multicast on an overlay network (Fig. \ref{fig:diagram}, 1)
among block \emph{producers}.  At each round, one or more producers propose
candidate blocks extending the blockchain history (Fig. \ref{fig:diagram}, 2).  A set of
\emph{validator} nodes (generally the same set or a superset of the producers) 
validates and selects one of the blocks through
a consensus mechanism (Fig. \ref{fig:diagram}, 3).  \SPEEDEX{} is suitable for integration into a
variety of blockchains, but benefits from a consensus layer with
relatively low latency (on the order of seconds), such as
BA$\star$~\cite{gilad:algorand}, SCP~\cite{lokhava:stellar}, or
HotStuff~\cite{yin:hotstuff}.


\begin{figure}  
\centering

\includegraphics[width=\columnwidth]{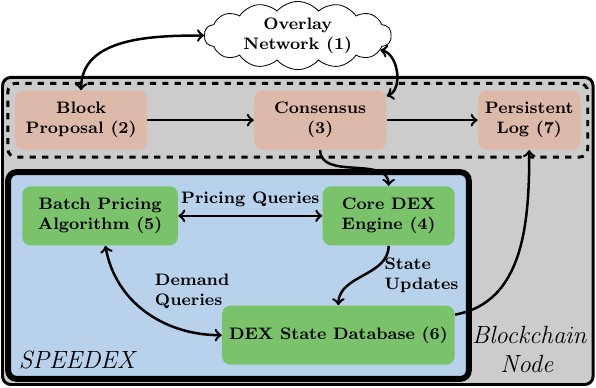}

\caption{
    Architecture of \SPEEDEX{} module (4, 5, 6) inside one blockchain node.
    \label{fig:diagram}
}

\end{figure}

The implementation evaluated here uses HotStuff~\cite{yin:hotstuff},
while the \thestellarbc{} implementation relies on \stellar{}'s
existing consensus protocol, SCP~\cite{mazieres2015stellar}.

Most central banks and digital currency issuers maintain a
ledger tracking their currency holdings.  \SPEEDEX{} is not intended to
replace these primary ledgers.  Rather, we expect banks and
other regulated financial institutions to issue 1:1 backed token
deposits onto a blockchain that runs \SPEEDEX{} and provide interfaces for moving
money on and off the exchange. These assets could be digital-native
tokens as well; any divisible and fungible asset can integrate with \SPEEDEX.

\SPEEDEX{} supports four operations: account creation, offer creation,
offer cancellation, and send payment.  
Offers on \SPEEDEX{} are traditional limit orders.  For example, one offer might 
offer to sell 100~EUR to buy USD, at a price no lower than 0.91 USD/EUR.  Offers can
trade between any pair of assets, in either direction.  Another offer, for example,
might offer to sell 100~USD in exchange for EUR, at a price no lower
than 1.10~EUR/USD.

What makes \SPEEDEX{} different from existing DEXes is the manner in which
it processes new orders.  Traditional exchanges process trades sequentially,
implicitly computing a matching between limit orders.
\SPEEDEX{}, by contrast, processes trades in batches (typically, one batch would consist of all of the limit orders
in one block of the blockchain).  

In a blockchain, all of the transactions in a block are appended
at the same clock time, so there is no reason \emph{a priori} why a DEX should pick
one ordering over another.
\SPEEDEX{}, by design, imposes no ordering whatsoever
between transactions in a block.  Side effects of a transaction are only visible
to other transactions in future blocks.

Logically, when the \SPEEDEX{} core engine
(Fig. \ref{fig:diagram}, 4) receives a finalized block of trades, it applies all of the 
trades at exactly the same time and computes an unordered set of state changes, which it passes
to its exchange state database (Fig. \ref{fig:diagram}, 6).  This database records
orderbooks and account balances, and 
is periodically written to the persistent log (Fig \ref{fig:diagram}, 7).

\subsection{\SPEEDEX{} Module Architecture}

To implement an exchange that operates replicably 
where trades in a block are not ordered relative to each other,
\SPEEDEX{} requires a
set of trading semantics such that operations {\it commute}.  

Traditional exchange semantics are far from commutative:  one offer to
buy an asset is matched with the lowest priced seller, and the next
offer to buy is matched against the second-lowest priced seller, and
so on.  Hence, every trade can occur at a slightly different exchange
rate.


Instead, to make trades commutative, \SPEEDEX{} computes in every block a
\emph{valuation}
$p_{\asset{A}}$ for every asset ${\asset{A}}$.  The units of $p_{\asset{A}}$ are
meaningless, and can be thought of as a fictional valuation asset that
exists only for the duration of a single block.  However, valuations
imply exchange rates between different assets---every sale of asset
${\asset{A}}$ for asset ${\asset{B}}$ occurs at a price of $p_{\asset{A}}/p_{\asset{B}}$.  Unlike traditional exchanges,
\SPEEDEX{} does not 
explicitly compute a matching between trade offers.
Instead, offers trade with a conceptual \emph{``auctioneer''} entity at these
exchange rates.  Trading becomes commutative because all trades in one asset pair occur at the same price.

The main algorithmic challenge is to compute valuations where the
exchange \emph{clears}---
i.e., the amount of each asset sold to
the auctioneer equals the amount bought from the auctioneer.

When the auctioneer sets exact clearing valuations, 
an offer trades fully
with the auctioneer if its limit price is strictly below the
auctioneer's exchange rate, and not at all if its limit price exceeds
the auctioneers rate.  When the limit price equals the exchange rate,
\SPEEDEX{} may execute the offer partially.  Note that an exchange is a zero-sum system;
as compared to sequential execution, some users may see better prices and some may see worse,
but \SPEEDEX{} guarantees that no user's price is worse than their
minimum limit price.

\begin{theorem}
\label{thm:valsexist}

Exact clearing valuations always exist.  These valuations are unique up to
rescaling.\footnote{And technical conditions (\S\ref{sec:uniquevalsexist}), e.g. everything
clears an empty market.}

\end{theorem}

Theorem \ref{thm:valsexist} is a restatement of a general theorem of Arrow-Debreu exchange market theory \cite{devanur2016rational} (\S \ref{sec:uniquevalsexist}).


Concretely, whenever the core \SPEEDEX{} engine (Fig \ref{fig:diagram}, 4) receives
a newly finalized block, one of its first actions is to query an algorithm
that computes clearing valuations (Fig \ref{fig:diagram}, 5).  It then uses the output of this algorithm to
compute the modifications to the exchange state (Fig \ref{fig:diagram}, 6).

As valuations that clear the market always exist for any set of limit orders, 
there is no adversarial input that \SPEEDEX{} cannot process.
And because these valuations are unique, \SPEEDEX{} operators do not have a strategic
choice between different sets of valuations.  \SPEEDEX{}'s algorithmic task is
to surface information about a fundamental mathematical property of a batch.


Unfortunately, we are not aware of a practical method to compute
clearing prices exactly.  (The number of bits required to represent
exact clearing prices may be extremely
large~\cite{devanur2016rational}, and in a natural extension of the
\SPEEDEX{} model~\cite{ramseyer2022batch} the clearing prices are not
even rational.)  \SPEEDEX{} therefore uses \emph{approximate} clearing
prices.

At nonexact clearing prices, the conceptual auctioneer will not have
enough of some asset(s) to pay out all offers willing to accept the
market price.  \SPEEDEX{} addresses this deficit in two ways.  First,
the auctioneer proportionally reduces the amount it pays out to offers
by a small fraction---in other words, it charges a commission.
Commissions are common for exchanges, whether decentralized or not,
though \SPEEDEX{} does it for market clearing rather than profit
reasons.  To avoid incentivizing high trading costs, the
implementation returns commissions to the asset issuers,
and one goal of our price computation algorithm's design 
is to make this commission as low as
computationally practical.
Second, the
auctioneer can refrain from filling some marketable offers.  Whereas
in a perfect Arrow-Debreu exchange market, offers at the market price
may be partially filled or not filled, in \SPEEDEX{} the same applies
to offers very close to the market price, even if they still beat the
market price by a small percentage.

\SPEEDEX{} always rounds trades in favor of the auctioneer.
Our implementation burns collected transaction fees and accumulated rounding
error
(effectively returning them to the issuer by reducing the issuer's
liabilities).  The \stellar{} implementation eliminates the fee
 and returns the accumulated rounding error to asset issuers.




\subsection{Design Properties}

\paragraph{Computational Scalability}

\SPEEDEX{}'s commutative semantics allow effective parallelization of
DEX operation.  Because transactions within a block are not
semantically ordered, DEX state is identical regardless of the order
in which transactions are applied.  This exact replicability is, of
course, required for a \emph{replicated} state machine.  The
order-independence also means \SPEEDEX{} transactions can be executed
in parallel by all available CPU cores despite the fact that thread
interleaving is nondeterministic in multicore machines.  Almost all
coordination occurs via hardware-level atomics (e.g., atomic add on
64-bit integers) without spinlocks.

\SPEEDEX{} stores balances in accounts, rather than in discrete,
unspent coins (often called ``UTXOs'').  It also supports
single-currency payment operations, which are simpler than DEX
trading.  Hence, \SPEEDEX{} disproves the popular
belief~\cite{opencbdc,reijsbergen2020exploiting}
that account-based ledgers are not compatible
with horizontal scalability.


\paragraph{No risk-free front running}

Well-placed agents in real-world financial markets can spy on
submitted offers, notice a new transaction $T$, 
and then submit a transaction $T'$ (that executes
before $T$) that buys an asset and re-sells it to $T$ at a slightly higher
price.  In some blockchain settings, $T'$ can be done as a
single atomic action \cite{daian2019flash}.  However, since every
transaction sees the same clearing prices in \SPEEDEX{}, back-to-back
buy and sell offers would simply cancel each other out.
Relatedly, because every offer sees the same
prices, a user who wishes to trade immediately can set a very low
minimum price and be all but guaranteed to have their trade executed,
but still at the current market price.

Risk-free front-running is one instance of the widely discussed 
``Miner Extractable Value'' (MEV) \cite{daian2019flash}
phenomenon, in which block producers reorder transactions
within a block for their own profit (or in exchange for kickbacks).
By eliminating the ordering of transactions
within a block, \SPEEDEX{} eliminates a large source of MEV.  However,
this does not eliminate every type of front-running manipulation, 
such as delaying victim transactions to a future block (see \S\ref{sec:limitations}).

\paragraph{No (internal) arbitrage and no central reserve currency}

An agent selling asset ${\asset{A}}$ in exchange for asset ${\asset{B}}$ will see a price
of $p_{\asset{A}}/p_{\asset{B}}$.  An agent trading ${\asset{A}}$ for ${\asset{B}}$ via some intermediary
asset ${\asset{C}}$ will see exactly the same price, as $\frac{p_{\asset{A}}}{p_{\asset{C}}} \cdot
\frac{p_{\asset{C}}}{p_{\asset{B}}}=\frac{p_{\asset{A}}}{p_{\asset{B}}}$.  Hence, one can efficiently trade between assets
without much pairwise liquidity with no need to search for an optimal
path.  By contrast, many international payments today go through USD
because of a lack of pairwise liquidity.  The multitude of USD-pegged
stablecoins in modern blockchains further fragments liquidity.
Of course, there can still be
arbitrage between \SPEEDEX{} and external markets.  

%% file: tatonnement.tex
\section{Commutative DEX Semantics}
\label{sec:commutativesemantics}

To propose or execute a block of transactions, the \SPEEDEX{} core engine performs the following three actions.

\begin{enumerate}[nolistsep]

    \item[1] For each transaction in the block (in parallel), check
      signature validity, collect new limit offers, and compute
      available account balances after funds are committed to offers
      or transferred between accounts.  
      When proposing a block of transactions, \SPEEDEX{}
      discards potentially invalid transations.

    \item[2] When proposing a block, 
    compute approximate clearing prices and approximation
      correction metadata.

    \item[3] Iterate over each offer, making a trade or adding it to the resting orderbooks
    (based on the prices and metadata).

\end{enumerate} 

For transaction processing in step 1 to be commutative, it must be the case that the step 1 output effects
(specifically: create
a new account, create a new offer, cancel an existing offer, and send a payment) of one 
transaction have no influence on the output effects of another transaction.  This means that one transaction
cannot read some value that was output by another transaction (in the same block), and that whether
one transaction succeeds cannot depend on the success of another transaction.


To meet the first requirement, traders include all parameters
to their transactions within the transaction itself.  
The second requirement necessitates precise management of
transaction side effects.  
At most one transaction per block may
alter an account's metadata (such as the account's public key
or existence), and metadata changes take effect only at the end of
block execution.
Similarly, an offer cannot be created and cancelled in the same block.
As payments and trading are the common case, we
do not consider these restrictions a serious limitation.

\SPEEDEX{} must also ensure that no account is overdrafted.  That is
to say, after processing all transactions in a block, the unlocked
balance of every account must be nonnegative (where an open offer
locks the offered amount of an asset for the duration of its
lifetime).  Unlike most distributed ledgers, \SPEEDEX{} cannot simply
deem the second of two conflicting transactions to fail---after all,
transactions have no ordering.  Instead, our implementation requires a
block proposer to ensure that a block cannot cause overdrafts; every
node rejects blocks that violate this property.  To generate valid
blocks, proposers use a conservative process outlined in \S
\ref{sec:nd_assembly}.  The design requires passing information from
the \SPEEDEX{} database (Fig \ref{fig:diagram}, 6) to the proposal
module (Fig \ref{fig:diagram}, 2).

The core remaining technical challenge is the batch price computation (Fig \ref{fig:diagram}, 5).

\section{Price Computation}

\subsection{Requirements}

As discussed earlier, in every block, \SPEEDEX{} computes batch clearing prices and executes trades in response to these prices.
Every DEX is subject to two fundamental constraints:

\begin{compactitem}
\item

{\bf Asset Conservation} No assets should be created out of nothing.  As discussed in
\S\ref{sec:architecture}, offers in \SPEEDEX{} trade with a virtual
auctioneer.  After a batch of trades, this auctioneer cannot be left
with any debt.  We do allow the auctioneer to burn some surplus assets
as a fee.

\item  {\bf Respect Offer Parameters}
No offer trades at a worse price than its limit price.
\end{compactitem}

Additionally, \SPEEDEX{} should facilitate as much trade volume as
possible.  (Otherwise, the constraints could be vacuously met by never
trading.)  Furthermore, price computation must be efficient, as it
occurs for each block of trades, every few seconds.  Finally,
\SPEEDEX{} should minimize the number of offers that trade partially;
asset quantities are stored as integer multiples of a minimum unit, so
each partial trade risks accumulating a rounding error of up to one
unit.

\subsection{From Theory To Practice}
\label{sec:theorytopractice}

The problem of computing batch clearing prices is equivalent to the problem
of computing equilibria in linear Arrow-Debreu Exchange Markets (\S \ref{sec:speedexandad}).  Our algorithm is based
on the iterative \Tat{} process from this literature \cite{codenotti2005market}.

However, the runtimes of the theoretical algorithms scale very poorly, both asymptotically and empirically.  They also output
approximate equilibria for notions of approximation that violate the two fundamental constraints above (for example,
Definition 1 of \cite{codenotti2005market} permits equilibria to mint new assets and to steal from a user).

We develop a novel algorithm for computing equilibria that runs efficiently in practice (\S \ref{sec:tateval})
and explicitly ensures that (1) asset amounts are conserved and (2) every offer trades at exactly the market prices,
and only if the offer's limit price is at or below the batch exchange rate.  First, \Tat{} approximates clearing prices (\S \ref{sec:tatonnement}).
We show that the structure of the types of trades in \SPEEDEX{}
lets each iteration run in time logarithmic in the number of open limit offers (via a series of binary searches), giving an algorithm
asymptotically faster than that within the theoretical literature.


We then explicitly correct for the approximation error with a linear program (\S \ref{apx:lp}).
Crucially, the size of this linear program is linear in the number of asset pairs, and has no dependence 
on the number of open trade offers.  The linear program ensures that, no matter what prices \Tat{} outputs, (1) asset amounts are conserved, and (2) no offer
trades if the batch price is less than its limit price.

To be precise, our algorithm outputs the following:

\begin{compactitem}

    \item {\bf Prices:} For each asset ${\asset{A}}$, \SPEEDEX{} computes an asset valuation $p_{\asset{A}}$.
    One unit of ${\asset{A}}$ trades for $p_{\asset{A}}/p_{\asset{B}}$ units of ${\asset{B}}$.

    \item {\bf Trade Amounts:}  For each asset pair $({\asset{A}}, {\asset{B}})$, \SPEEDEX{} computes an amount $x_{{\asset{A}},{\asset{B}}}$ of 
    asset ${\asset{A}}$
    that is sold for asset ${\asset{B}}$ (again, at exchange rate $p_{\asset{A}}/p_{\asset{B}}$).

\end{compactitem}

For every asset pair $({\asset{A}}, {\asset{B}})$, \SPEEDEX{} sorts all of the offers selling ${\asset{A}}$ for ${\asset{B}}$ by their limit prices,
and then executes the offers with the lowest limit prices, until it reaches a total amount of ${\asset{A}}$ 
sold of $x_{{\asset{A}},{\asset{B}}}$
(tiebreaking by account ID and offer ID).

As a bonus, this method ensures that at most one offer per trading pair executes partially, minimizing rounding error.



\section{Price Computation: \Tat{}}
\label{sec:tatonnement}

\Tat{} is an iterative process; starting from an (arbitrary) initial set of prices, it iteratively refines them until
the prices reach a stopping criterion.


Each iteration of \Tat{} starts with a \emph{demand query}.  The
\emph{demand} of an offer is the net trading of the offer 
(with the auctioneer) in response to a set of prices, and the demand of a set of offers is the sum of the demands of each offer.
\Tat{}'s goal is to find prices such
that the amount of each asset sold to the auctioneer matches the amount bought from it (in other words,
the net demand is $0$).

\begin{example}
Suppose that a limit order offers to sell 100 USD for EUR with a
minimum price of\/ $0.8$ EUR per USD.
    If the candidate prices are such that $\alpha =
    \frac{p_{\text{USD}}}{p_{\text{EUR}}} > 0.8$, then the limit order
    would like to trade,
and its demand is $(-100~ \text{USD}, 100\alpha~ \text{EUR})$.  Otherwise, its
demand is $(0~ \text{USD}, 0~ \text{EUR})$.
\end{example}


\paragraph{Iterative Price Adjustment.}
If more units of an asset are demanded from the auctioneer than are supplied to it (a positive net demand, meaning a deficit for the auctioneer),
then the auctioneer raises the price of the asset.  
Otherwise, the auctioneer has a surplus, so 
it lowers the price of the asset.
Implementing this process effectively requires careful numerical normalization in response to 
differences in prices and trade volumes, which we describe in detail in \S \ref{sec:priceupdate}.


\Tat{} repeats this process until the current set of prices is sufficiently close to 
the market clearing prices (or it hits a timeout).  Specifically, \Tat{} iterates until it 
has a set of prices such that, if the auctioneer charges
a commission of $\varepsilon$, then there is a way to execute offers such that:

\begin{compactenum}
\item [1] The auctioneer has no deficits (assets are conserved)
\item [2] No offer executes outside its limit price bound
\item [3] Every offer with a limit price more than a $(1-\mu)$ factor below the auctioneer's exchange rate executes in full.
\end{compactenum}

The last condition is a formalization of the notion that \SPEEDEX{} should satisfy as many trade requests as possible.  Informally,
an offer with a limit price equal to the auctioneer's exchange rate is indifferent between trading and not trading,
while one with a limit price far below the auctioneer's exchange rate strongly prefers trading to not trading.



\subsection{Efficient Demand Queries}
\label{sec:logtransform}

Implemented na\"\i vely, \Tat{}'s demand queries would consist of a loop over every open exchange offer.  
This is impossibly expensive, even if the loop is massively parallelized.
Concretely, one invocation of \Tat{} can require many thousands of demand queries.  Every demand query therefore
must return results in at most a few hundred microseconds. 

This na\"\i ve loop appears to be required for the (more general) problem instances studied in the theoretical literature.
However, all of the offers in \SPEEDEX{} are traditional limit orders
that sell one asset in exchange for one other asset at some limit price.
An offer with a lower limit price always trades if an offer with a higher limit price trades.
Therefore, \SPEEDEX{} groups offers by asset pair and sorts offers by their limit prices.
We drive the marginal cost of this sorting to near zero by using an offer's limit price
as the leading bits (in big-endian) of the keys in our Merkle tries (\S\ref{sec:fastsorting}).

\SPEEDEX{} can therefore compute a demand query with a sequence of binary searches
(\S \ref{sec:tatpreprocessing}).
Individual binary searches can run on separate CPU cores.
The number of open offers (say, $M$) on an exchange is vastly higher than the number of assets traded (say, $N$).
Our experiments in \S \ref{sec:scalable} trade $N=50$ assets with $M=$ tens of millions of open offers; the complexity reduction
from $O(M)$ to $O(N^2\lg(M))$ is crucial.

\subsection{Multiple \Tat{} Instances}
\label{sec:multiinstance}
\S \ref{apx:tatmods} describes several other \Tat{} adjustments that help it respond well to a wide variety of market conditions.
Some of these adjustments are parametrized (such as how quickly one
should adjust the candidate prices); rather than pick one set of
control parameters, we run several instances of \Tat{} in parallel and
take whichever finishes first as the result.  (In the case of a
timeout, we choose the set of prices that minimizes the
\emph{unrealized utility} [\S\ref{sec:imbalanced}]\@.)  \SPEEDEX{}
includes the output of \Tat{} and the subsequent linear program in the
headers of proposed blocks (\S \ref{sec:header}).



%% file: tatonnement_evaluation.tex
\section{Evaluation: Price Computation}
\label{sec:tateval}

\Tat{}'s runtime depends primarily on the target approximation accuracy,
the number of open trade offers, and the distribution of the open
trade offers.  The runtime increases as the desired accuracy
increases.  Surprisingly, the runtime actually {\it decreases} as the number
of open offers increases.  And like many optimization problems, \Tat{} performs best when
the input is normalized, meaning in this case that the (normalized, \S \ref{sec:priceupdate}) 
volume traded of each asset 
is roughly the same.

\Tat{} runs once per block.  
To produce a block every few seconds, \Tat{} must run 
in under one second most of the time.  Our implementation runs \Tat{}
with a timeout of 2 seconds, but it typically converges much faster.

\subsection{Accuracy and Orderbook Size}

We find that \Tat{} converges more quickly as the number of open offers {\it increases}.
\Tat{} converges fastest when small price changes do not cause
comparatively large changes in overall net demand.
However, an offer's behavior is a discontinuous function (of prices);
it does not trade below its limit price and trades fully above it.

There are two factors that mitigate these ``jump discontinuities.''
First, \Tat{} approximates optimal offer behavior by a continuous function (\S \ref{apx:error}).
Smaller $\mu$ means a closer approximation.
Second, the more offers there are in a batch, the smaller any one offer's relative contribution to overall demand.
This last factor explains why \Tat{} converges more quickly when there are more offers on the exchange.
A real-world deployment might raise accuracy as trading increases.

Fig. \ref{fig:mintxs} plots the
minimum number of trade offers that \Tat{} needs to consistently find
clearing prices for $50$ distinct assets in under $0.25$ seconds
(for the same trade distribution used in \S \ref{sec:scalable}).
To put these fee rates in context, BinanceDex \cite{binancedexfees} charged a fee of either
$0.1\%\approx 2^{-10}$ or $0.04\%\approx 2^{-11.3}$.
Uniswap~\cite{uniswapv2,uniswapv3} charges $1\%$, $0.3\%$, or
$0.05\%$ (${\sim}2^{-6.6}$, ${\sim}2^{-8.4}$, and
${\sim}2^{-11}$, respectively), and Coinbase charges $0.5\%$ to
$4\%$~\cite{coinbasefees} (
$\sim 2^{-7.6}$ to $\sim 2^{-4.6}$).

\begin{figure}  
\centering
\includegraphics[width=\columnwidth]{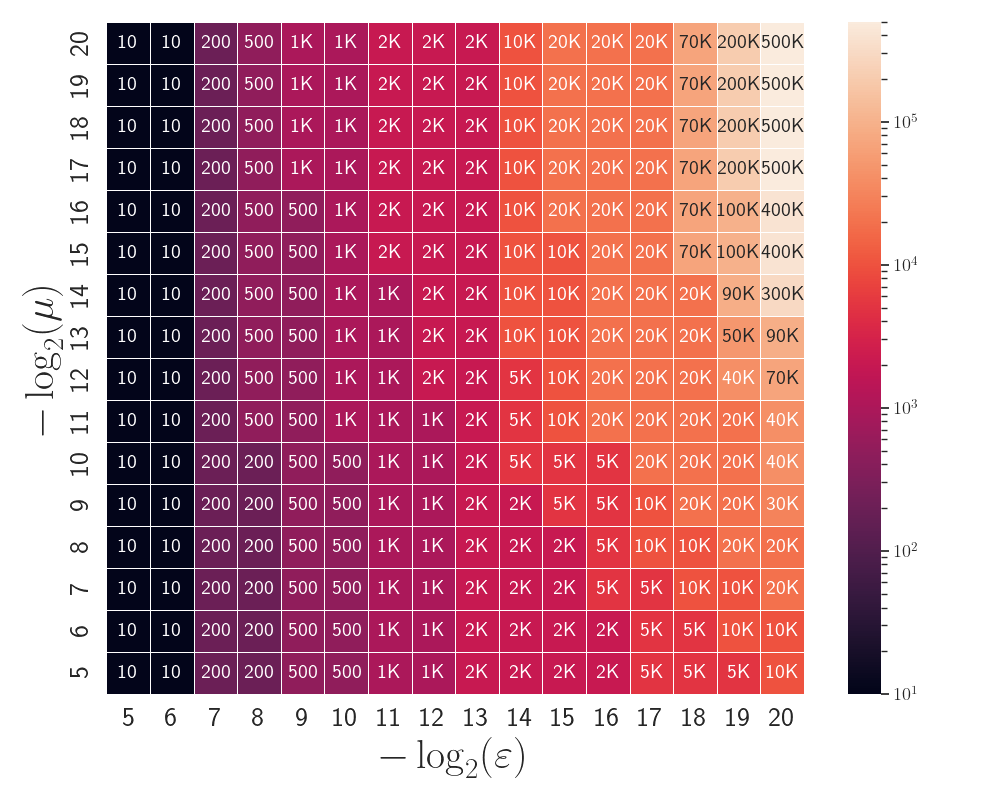}

\caption{
    Minimum number of offers needed for \Tat{} to run in under $0.25$ seconds (Smaller is better. Times averaged over 5 runs).
    The x axis denotes offer behavior approximation quality ($\mu$), 
    and the y axis denotes the commission ($\varepsilon$).
    \label{fig:mintxs}
}

\end{figure}


Though our experiments rarely experienced \Tat{} timeouts, \Tat{}
timeouts caused by sparse orderbooks may be self-correcting:  If
\SPEEDEX{} proposes suboptimal prices, fewer offers will find a
counterparty and trade.  When fewer offers clear in one block, more
are left to facilitate \Tat{} in the next block.
\S \ref{sec:alternatestrats} describes an alternative algorithm that is effective on small batches.


\subsection{Robustness Checks}
\label{sec:imbalanced}


As a robustness check, we run \Tat{} against a trade distribution
derived from volatile cryptocurrency market data.  In an ideal
world, we could replay trades from another DEX through \SPEEDEX{}\@.
Unfortunately, doing so poses several problems.
First, in practice, almost all DEX trades go through four de facto
reserve currencies (ETH, USD, USDC, and USDT), three of which are
always worth close to \$1.  The decomposition between a few core
``pricing'' assets and a larger number of other assets 
makes price discovery
too simple. 
Second,
transaction rates on existing DEXes are too low to provide enough data. 
Finally, we suspect users would submit different
orders to \SPEEDEX{} than they might on a traditional exchange,
due to the distinct economic properties of batch trading systems.

\paragraph{Experiment Setup.} 
As a next-best alternative, we generate a dataset based on
historical price and market volume data.  We took the 50 crypto assets
that had the largest market volume on December 8, 2021 (as reported by
\url{coingecko.com}) and for each asset, gathered 500 days of price
and trade volume history.  We then generated 500 batches of 50,000
transactions.  A new offer in batch $i$ sells asset $\asset{A}$ (and buys
asset $\asset{B}$) with probability proportional to the relative volume of
asset $\asset{A}$ (and asset $\asset{B}$, conditioned on $\asset{A}\neq \asset{B}$) on day $i$, and
demands a minimum price close to the real-world exchange rate on day
$i$.
The extreme volatility of cryptocurrency markets and variation between
these 50 assets make this dataset particularly difficult for \Tat{}.  To
further challenge \Tat{}, we use a smaller block size of $\sim30,000$
(compared to $500,000$ in \S\ref{sec:scalable}).


The experiment charged a commission of
$\varepsilon=2^{-15}\approx 0.003\%$, and attempted to clear offers
with limit prices more than $1-\mu$ below the market prices, for
$\mu=2^{-10}\approx 0.1\%$ (\S \ref{apx:error}).

\paragraph{Experiment Results.}
The experiment ran for 500 blocks.  
Each block created about 25,000 new offers and a few thousand
cancellations and payments. 

\Tat{} computed an equilibrium quickly
in 350 blocks, and in the remainder, computed prices sufficiently close
to equilibrium that the follow-up linear program facilitated the vast
majority of possible trading activity.


We measure the quality of an approximate set of prices by the ratio of the ``unrealized utility'' to the ``realized utility.''
The utility gained by a trader from selling one unit of an asset is the difference between the market exchange rate and the trader's limit price,
weighted by the valuation of the asset being sold.  Note that the units do not matter when comparing relative amounts of ``utility.''

In the blocks where \Tat{} computed an equilibrium quickly, the mean ratio of unrealized to realized utility was $0.71\%$
(max: $4.7\%$), and in the other blocks, the mean ratio was $0.42\%$ (max: $3.8\%$).

Recall that \Tat{} terminates as soon as a stopping criteria is met;
roughly, ``does the supply of every asset exceed demand,''
so one mispriced asset will cause \Tat{} to keep running.
However, every \Tat{} iteration continues to refine the price of every asset.  This is why
\Tat{} actually gives more accurate results in the batches it found challenging.  A deployment might enforce a minimum number of \Tat{} rounds.

Qualitatively, \Tat{} correctly prices assets with high trading volume
and struggles on sparsely traded assets (as might be expected from
Fig. \ref{fig:mintxs}).  \Tat{} also adjusts its price adjustment rule
in response to recent market conditions (\S \ref{sec:priceupdate}), 
a tactic which is less effective on volatile assets.

Should this pose a
problem in practice, a deployment could choose to vary the approximation parameters by
trading pair.

%% file: scalability_evaluation.tex
\section{Evaluation: Scalability}
\label{sec:scalable}

We ran \SPEEDEX{} on four r6id.24xlarge instances in an Amazon Web Services datacenter.
Each instance has 48 physical CPU cores divided over two Intel Xeon Platinum 8375CL chips (32 total cores per socket, 24 of which are allocated to
our instances), running at
2.90Ghz with 
hyperthreading enabled, 768GB of memory, 4 1425GB NVMe drives connected in a RAID0 configuration.  
We use the XFS filesystem \cite{xfs}.  
These experiments use the HotStuff consensus protocol \cite{yin:hotstuff},
 and do not include Byzantine replicas or a rotating leader.

\paragraph{Experiment Setup.}
These experiments simulate trading of 50 assets.
Transactions are charged a fee of $\varepsilon=2^{-15}(0.003\%)$.  We
set $\mu=2^{-10}$, guaranteeing full execution of all orders
priced below $0.999$ times the auctioneer's price.  The initial database contains 10 million accounts.  \Tat{} never timed out,
and typically required fewer than 1,000 iterations.

Transactions are generated according to a synthetic data model---every
set of 100,000 transactions is generated as though the assets have
some underlying valuations, and users trade a random asset pair using
a minimum price close to the underlying valuation ratio.  The
valuations are modified (via a geometric Brownian motion) after every set.
Accounts are drawn from a power-law distribution. \XXX{say param = 0.0000001?}

Each set is split into four pieces, with one piece given to each replica.  Replicas load these sets sequentially
 and broadcast each set to every other replica.  
Each replica adds received transactions to its pool of unconfirmed transactions.

Replicas propose blocks of roughly 500,000 transactions.
In these experiments, each block consists of roughly 350,000--400,000
new offers, 100,000--150,000 cancellations, 10,000--20,000 payments,
and a small number of new accounts.  We generate 5,000 sets of input
transactions.  Some of these transactions conflict with each other and
are discarded by \SPEEDEX{} replicas.  Each experiment runs
for 700--750 blocks.

Every five blocks, 
the exchange commits its state to persistent storage in the background (via LMDB \cite{lmdb}, \S \ref{sec:persistence}).

\paragraph{Performance Measurements.}
Fig.~\ref{fig:e2e} plots the end-to-end transaction throughput rate of
\SPEEDEX{} as the number of worker threads inside \SPEEDEX{} increases.  
The x-axis plots the number of open offers on the exchange.

\begin{figure}  
\centering
\includegraphics[width=\columnwidth]{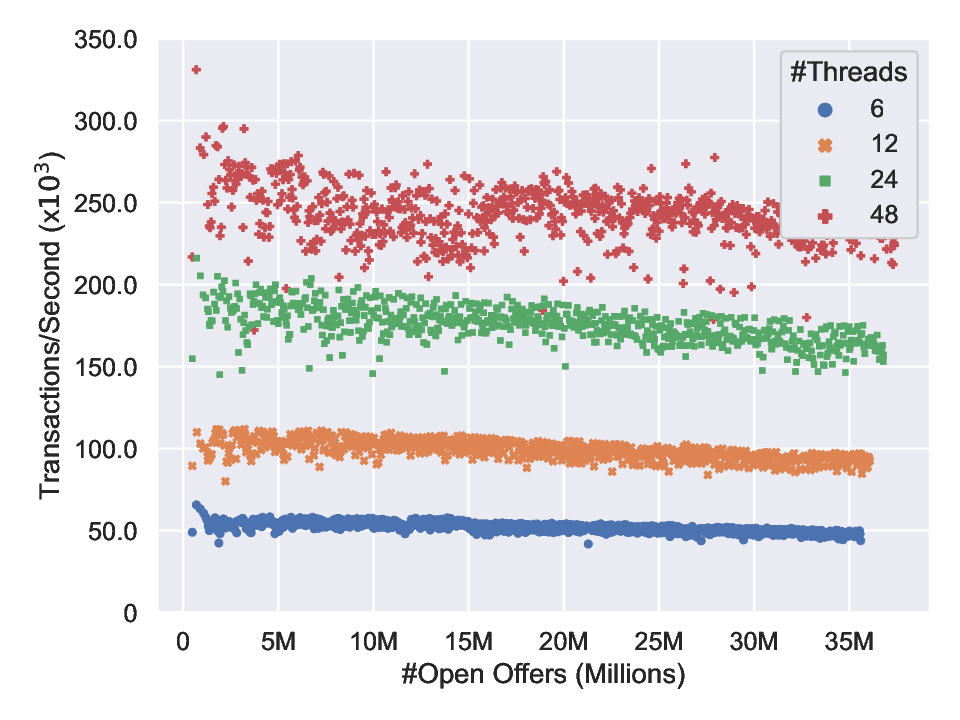}
\caption{
    Transactions per second on \SPEEDEX{}, plotted over the number of open offers.
    \label{fig:e2e}
}
\end{figure}

Most importantly, Fig.~\ref{fig:e2e} demonstrates that \SPEEDEX{} can efficiently use its available CPU hardware.  The speedup 
is near-linear, until the number of threads approaches the number of CPU cores---from 6 to 12, $\sim 1.9x$, from 12 to 24, $\sim 1.8x$, and from 24
to 48, $\sim 1.4x$.  The thread counts are only for the number of threads directly for \SPEEDEX{}'s critical path, and not for many of the tasks that the implementation must perform in the background, such as logging data to persistent storage (logging the account database uses 16 threads), consensus, and garbage collection, and these threads
begin to contend with \SPEEDEX{} as the number of \SPEEDEX{} worker threads increases.

Secondly, Fig.~\ref{fig:e2e} demonstrates the scalability of \SPEEDEX{} with respect to the number of open offers.  
The number of open offers \SPEEDEX{}
works with in these experiments is already quite large, but most importantly, as the number of open offers goes from 0 to the 10s of millions,
\SPEEDEX{}'s transaction throughput falls by only $\sim 10\%$.  
This slowdown is primarily derived from a \Tat{} optimization (the precomputation
outlined in \S \ref{sec:tatimpldetails}).  \Tat{} is the one part of \SPEEDEX{} that 
cannot be arbitrarily parallelized, so we design our implementation towards making it as fast as possible.  
An implementation might skip this work in some parameter regimes.

To focus on the performance of \SPEEDEX{}, Figs. \ref{fig:prodscale}
and \ref{fig:valscale} plot the time to propose and execute blocks,
and to validate and execute proposals, respectively, when we disable signature
verification (which is trivial to parallelize).
First, note that both proposal and validation scale with the number of threads; validation scales better
than proposal due to the aforementioned \Tat{} optimization. Second, note that validating and executing a proposal from another replica is substantially faster
than proposing a block; this lets a replica that is somehow delayed catch up.

The runtime variation in Fig.~\ref{fig:prodscale} results from the
fact that \SPEEDEX{} without signature verification runs too quickly for our
persistent logging implementation.

\begin{figure}  
\centering
\includegraphics[width=\columnwidth]{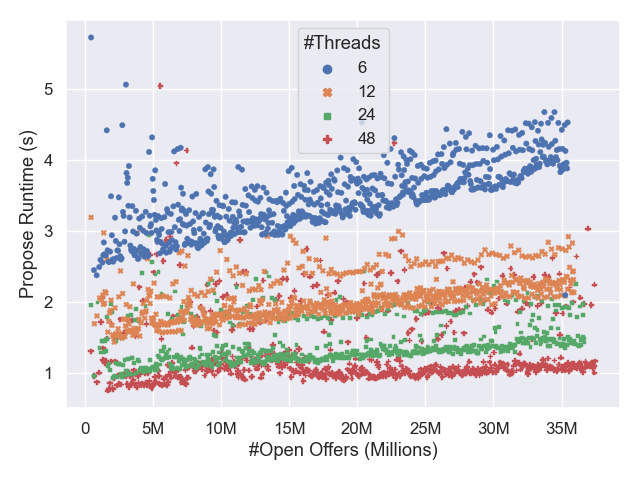}

\caption{
  	Time to propose and execute a block, plotted over the number of open offers.
  	\label{fig:prodscale}
}

\end{figure}


\begin{figure}  
\centering
\includegraphics[width=\columnwidth]{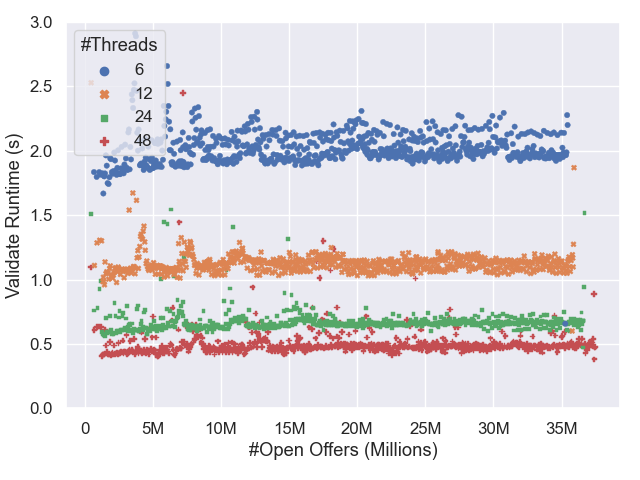}

\caption{
  	Time to validate and execute a proposal, plotted over the number of open offers (measurements from one replica).
  	\label{fig:valscale}
}
\end{figure}

\SPEEDEX{} is not a consensus protocol, and these experiments (one consensus invocation
every few seconds) do not come close to stressing the consensus throughput of Hotstuff.
However, network bandwidth requirements necessarily scale (at least) linearly with transaction rate.
Recent work, such as \cite{danezis2022narwhal,keidar2021all,yang2022dispersed},
develops consensus protocols that maximally use available network bandwidth.  
However, integrating \SPEEDEX{} with any consensus protocol requires understanding
the tradeoffs between batch size, transaction rate, and consensus frequency.  Fig.~\ref{fig:blocksz}
plots this tradeoff running \SPEEDEX{} on the same transaction workload as in Fig.~\ref{fig:e2e}.
We also ran \SPEEDEX{} with more replicas on different hardware
and observed the same scalability trends, as outlined in \S
\ref{sec:10rep} (albeit with lower overall throughput on weaker hardware).

\begin{figure}  
\centering
\includegraphics[width=\columnwidth]{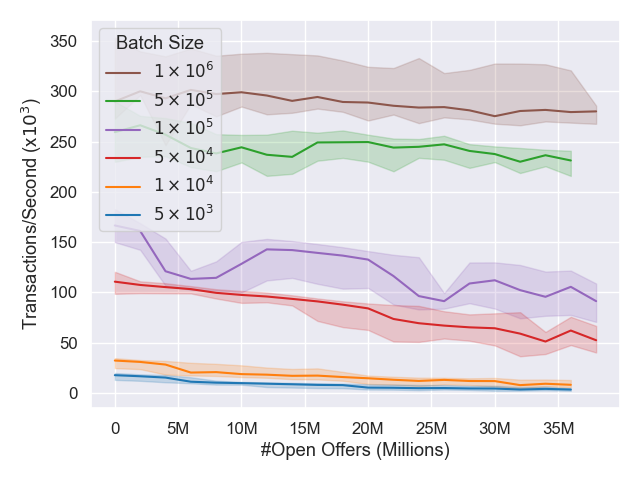}

\caption{
    Median transaction rates, varying block size and number of open offers (grouped into buckets of 2M).  Shaded areas plot 10th to 90th percentiles.
    \label{fig:blocksz}
}
\end{figure}


\paragraph{Conclusions.}
To reiterate, \SPEEDEX{} achieves these transaction rates
while operating fully on-chain,
with no offchain rollups and no sharding of the exchange's
state.  To make \SPEEDEX{} faster, one can simply give it
more CPU cores, without changing the transaction semantics or user
interface.  This scaling property is unique among existing DEXes.

\subsection{Alternative Scaling Techniques}

\paragraph{Traditional Exchange Semantics.}
The core logic of just an exchange system can be implemented extremely efficiently with almost no code.
The logic of the constant product market maker UniswapV2 \cite{uniswapv2},
for example, is less than 10 lines of simple arithmetic code.  
An orderbook-based
exchange requires more code but can still be made very fast, as most operations modify only a small number of 
data objects.  We implemented a bare-bones orderbook exchange with two assets using the same data structures as in
\SPEEDEX{}---each transaction checks the orderbook for a matching
offer or offers and either makes appropriate transfers or adds
the new offer to the orderbook.  These operations are extremely fast
when the number of accounts is small; our implementation
runs $\sim 1.7$~million of these transactions per second when there are only 100 accounts.  However, every database lookup becomes slower
as the as the number of accounts grows; when there are 10 million accounts in the database (as in the above experiments),
throughput falls 8x to $\sim 210,000$ per second.  
Yet that is before adding all of the other \SPEEDEX{} features
one needs in a real DEX, such as state hashes, transaction fees,
 structures for simple payment verification \cite{nakamoto:bitcoin}, replication, or durable logging.
The scalability of the full \SPEEDEX{} implementation lets it surpass that rate even when slowed down by all of these features.

Note that every orderbook operation affects every subsequent
transaction---each transaction influences
the exchange rate observed in the next transaction---and as such, 
their execution cannot be parallelized.  \SPEEDEX{}'s design, therefore, enables parallel execution of what would otherwise
be a strictly serial workload.
To isolate the effect of \SPEEDEX{}'s parallelizable semantics on its transaction throughput, we therefore turn to
a workload that does not touch the DEX at all---one where every transaction is a payment between random accounts.

\paragraph{Optimistic Concurrency Control.}
A widely explored class of alternative designs for parallel transaction execution use optimistic concurrency control,
and of these approaches the most closely related state of the art design appears to be Block-STM \cite{gelashvili2022block},
which is deployed in Aptos \cite{aptoswhitepaper}.
This approach
optimistically executes batches of transactions, retrying after conflicts as necessary.


We therefore design the measurements of Fig.~\ref{fig:blockstm} to mirror the experiments in \cite{gelashvili2022block}.
The ``Aptos p2p'' transactions
in \cite{gelashvili2022block} are payments between two random accounts, and consist of 8 reads of 5 writes.
Each of our payments consists of two data reads (source account public key and last committed sequence number), two atomic compare\_exchange operations (subtract payment and fee from source), an atomic fetch\_xor (reserve sequence number), and an atomic fetch\_add (add payment to destination)---implemented without atomics, this would be 6 reads and 4 writes.  All payments are of the same asset.


\begin{figure}
\centering
\includegraphics[width=\columnwidth]{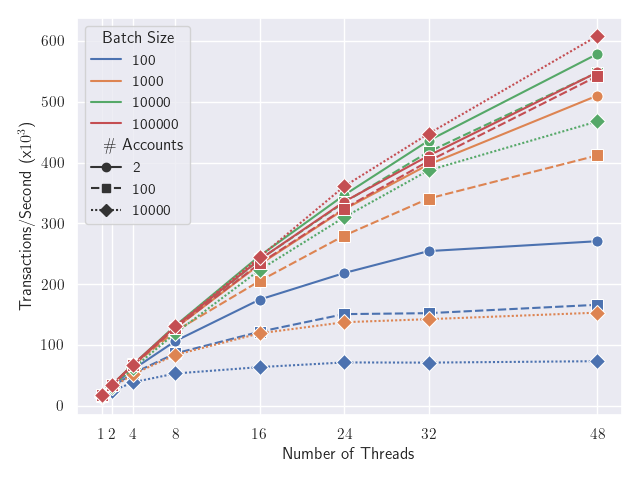}

\caption{
    Throughput of \SPEEDEX{} on batches of payment transactions with varying thread counts
 (average of 100 trials).
    \label{fig:blockstm}
}

\end{figure}

Fig.~\ref{fig:blockstm} plots the throughput rates of \SPEEDEX{} on this transaction workload for the parameter settings measured 
in Block-STM (Figs. 7 and 8, \cite{gelashvili2022block}).  Note that for large batch sizes,
the transaction throughput is largely independent of the number of accounts,
even though every transaction in the two account setting contends with every other transaction.  Furthermore, unlike Block-STM, \SPEEDEX{} achieves near-linear scalability on sufficiently large batches.
For small batch sizes, a large number of accounts actually slows down \SPEEDEX{}, largely due to increased sensitivity to cache performance and 
our system's NUMA (two socket) architecture on small timescales.  We also ran this experiment on a single-socket system (an AWS c5a.16xlarge,
as in \cite{gelashvili2022block}), and found only negligle impact of the number of accounts on throughput.
Fig.~\ref{fig:blockstm} was run with hyperthreading disabled, to compare against Block-STM experiments.
The rest of our experiments were run with hyperthreading enabled (because of the many background tasks in \SPEEDEX{});
enabling hyperthreading on this payments workload causes a negligible performance degradation for large batches (approximately 1-6\%),
and a larger (up to 25\%) on small batches.
As a baseline, \S \ref{sec:bstm_baseline} graphs the performance of Block-STM on these parameter settings on our hardware.

We also ran \SPEEDEX{} on an only-payments workload with 10 million accounts and 50 assets, and measured a throughput of approximately 375k, 215k, 114k, and 60k transactions per second using 48, 24, 12, and 6 threads, respectively (a 34.8x, 20.0x, and 10.6x, and 5.6x speedup over the single-threaded measurement).  
We disabled data persistence for these trials---again, the logging off of the critical path contends 
with \SPEEDEX{} at these transaction rates, especially for payment transactions that modify
two accounts, instead of just one (as when creating an offer).  
The throughput reached 255k transactions per second with data persistence enabled.

\paragraph{Production Systems.}
Finally, we ran the Ethereum Virtual Machine 
(Geth 1.10 \cite{geth}) 
on a workload of UniswapV2 \cite{uniswapv2} transactions, and measured a rate of $\sim 3000$ transactions per second 
(a result in line
with other Ethereum benchmarks \cite{wang2020prism}).  
The Loopring exchange, built as an L2 rollup on Ethereum,
claims a maximum rate of $\sim 2000$ per second \cite{loopringmainpage}, a number calculated
from Ethereum's per-block computation limit \cite{loopringdesigndoc}, which is in turn set based on 
the real computational cost of serial transaction execution \cite{perez2019broken,vitalikgaslimit,chen2017adaptive,aldweesh2021opbench}.
Precise measurements of the Stellar blockchain's orderbook DEX
\XXX{I don't think it's deanonymizing to mention it by name here}
suggest that its implementation
could handle $\sim 4000$ DEX trades per second.

%% file: limitations.tex
\section{Design Limitations and Mitigations}
\label{sec:limitations}

\paragraph{Latency.}
Batch trading inherently introduces latency (between order submission and order execution) not present
on traditional, centralized exchanges, simply because an order cannot execute until a batch has been closed
and clearing prices have been computed.  This latency is already present in a blockchain context (a transaction
is not finalized until the consensus protocol adds it to a block), so in this context, \SPEEDEX{} introduces
no additional latency. 

The latency may have downstream economic effects.  Market-making may
be more (or less) profitable
operating in a batch system, which could lead to reduced (or increased) liquidity.  
Budish et al.~\cite{aquilina2020quantifying}
argue that batch trading (between 2 assets) would reduce costs for market-makers, which could lead to increased liquidity.  However,
they study a higher batch frequency (approximately once per millisecond);
our lower batch frequency is less studied (see Q9, \cite{budishpubliccomment}).

\paragraph{\Tat{} Nondeterminism.}
The algorithms evaluated in \S \ref{sec:tateval}
can be viewed as a randomized approximation scheme, which raises the question of whether a malicious operator
can manipulate the approximation.  Note that the level of approximation error (as defined in \S \ref{apx:error}) can be measured,
so non-anonymous node operators can be penalized for malfeasance.
When regulation is not possible, \Tat{} can be made deterministic by fixing a set of control parameters for each instance
and choosing the solution with the lowest approximation error (or
lowest unrealized utility, \S \ref{sec:imbalanced}).
The Stellar implementation
uses a static set of control parameters with one \Tat{} instance.
Node operators could also
compete to compute prices accurately, as in \cite{dfusion}.

\paragraph{Nondeterministic Overdraft Prevention.}
\SPEEDEX{} needs to prevent an account from spending more than its balance of an asset.
As discussed in \S \ref{sec:commutativesemantics},
our implementation considers a proposal valid only if no account is overdrafted
after applying the block.  This design complicates pipelining of consensus with execution,
gives plausible deniability for delaying transactions, and is incompatible with cryptographic commit-reveal schemes.

Instead, given a fixed block of transactions,
an implementation could first compute, for each account, the total amount of each asset
debited from the account (before applying any credits).  If there is any possibility for an account
to overdraft in this block, then this amount must exceed the account's balance.
As such, to ensure that no accounts overdraft, the implementation can remove all transactions from accounts that might overdraft.
Note that this determination is made on a per-account basis, before any transactions are removed,
so this filtering requires only one, parallelizable pass over a block of transactions,
adding only minimal overhead (\S \ref{sec:filtering}).  Furthermore, only accounts
that attempt to overdraft are affected.  

Other commutativity conflicts, such
as cancelling an offer twice or reusing a sequence number, can be handled similarly,
by removing all transactions involved in these conflicts.  Note that using these filtering criteria,
removing a transaction cannot cause a commutativity conflict. 
\Thestellarbc{} plans this approach. 


\paragraph{Other Types of Front-Running.}
The set of pending transactions is public in many blockchains.  
One might estimate the clearing prices
in a future batch and arbitrage the batch against low-latency markets.  This could lead to negative externalities
(see \cite{budish2014implementation}, footnote 1), and could merit combining \SPEEDEX{} with a commit-reveal scheme
such as \cite{zhang2022flash,clineclockwork}.  
Such a design requires the deterministic overdraft-prevention scheme above.

Malicious nodes might also delay transactions.  
An implementation could buffer several blocks of transactions from a consensus protocol
into a single \SPEEDEX{} batch.  If even one of these consensus blocks is from
an honest replica (that does not censor transactions), a user could ensure that their transaction 
cannot be delayed from one \SPEEDEX{} batch to the next (by broadcasting to all replicas).  
This requires a consensus protocol
with sufficient \textit{chain quality} \cite{garay2015bitcoin}.
Alternatively, some DAG-based protocols \cite{danezis2022narwhal,keidar2021all}
simultaneously commit many blocks of transactions from different replicas.
Grouping these blocks into one \SPEEDEX{} batch, instead of ordering them arbitrarily,
achieves the same censorship-resistance property.
These designs would likewise require the deterministic overdraft-prevention scheme.


\paragraph{Linear Program Scalability.}
The runtime to solve the linear program increases dramatically beyond 60-80 assets,
limiting the number of assets in a \SPEEDEX{} batch.
A deployment could take advantage of market structure---there are many
assets (e.g., stocks) in the real world, but most are
linked to one geographic area or economy, and are primarily
traded against one currency. 
We formally show in \S \ref{sec:decomposition} that in this case, 
the price computation 
problem can
be decomposed between core pricing (i.e., numeraire) currencies and the
external stocks.  
After running \Tat{} on the core currencies, each stock can be priced
on its own relative to a core currency.
This lets \SPEEDEX{} support real-world transaction patterns 
with an arbitrary number of assets and a small number of pricing currencies.

\S \ref{apx:lp} points out that setting the commission to $0$ simplifies the linear program
to one that is more algorithmically tractable at larger numbers of assets.
The Stellar implementation uses this version of the linear program.

\paragraph{Limited Trade Types.}
Trades on \SPEEDEX{} are limited to trades selling a fixed amount of one asset in 
exchange for as much as possible of another.
\SPEEDEX{} does not implement offers to buy a fixed amount of an asset
in exchange for as little as possible of another.
These buy offers admit the same logarithmic transformation as in \S \ref{sec:logtransform},
but make the price computation
 problem PPAD-hard, a complexity class that is widely conjectured to
 be algorithmically intractable in polynomial time (\S \ref{sec:ppad}).
One could compute prices using only sell offers and integrate buy offers in the linear programming step.

Ramseyer et al.~\cite{ramseyer2022batch} show how to integrate
Constant Function Market Makers (CFMMs) \cite{angeris2019analysis}
into the exchange market framework and \Tat{}.
The Stellar implementation uses this integration with its own CFMMs.

%% file: system_design.tex
\section{Implementation Details} 

The standalone \SPEEDEX{} evaluated in \S\ref{sec:tateval} and
\S\ref{sec:scalable} is a blockchain using
HotStuff~\cite{yin:hotstuff} for consensus.  
A leader node periodically mints a new block from the memory pool and feeds
the block to the consensus algorithm.  Other nodes apply the block
once it has been finalized by consensus. 
A faulty node can propose an invalid block.
Consensus may finalize invalid blocks, but these blocks have no effect when applied.

The implementation is 
available open source at \speedexgithub{} and consists of $\sim$30,000 lines of
C++20, plus $\sim$5,000 lines for our Hotstuff implementation.
It uses Intel's TBB library~\cite{tbb} to manage parallel work scheduling,
the GNU Linear Programming
Kit~\cite{makhorin2008glpk} to solve linear programs, and
LMDB~\cite{lmdb} to manage data persistence (for crash recovery).

Exchange state is stored in a collection of custom
Merkle-Patricia tries; hashable tries allow nodes to efficiently 
compare state (to check consensus) 
and build short state proofs.


The rest of this section outlines additional design choices built
into \SPEEDEX{}.  Additional design choices in \S \ref{sec:additional_implementation}.
All optimizations (save \S\ref{sec:integration}) 
are implemented in the evaluated system.

\subsection{Blockchain Integration}
\label{sec:integration}

An existing blockchain with its own (non-commutative) semantics can integrate
\SPEEDEX{} by splitting block
execution into phases: first applying all \SPEEDEX{} transactions
(in parallel), then applying legacy transactions (sequentially).
\SPEEDEX{}'s scalability lets a blockchain charge only a marginal
 fee for transactions (to prevent spam).  A proof-of-stake integration of
\SPEEDEX{} could penalize faulty proposals.

\SPEEDEX{}'s economic properties are desirable independent of
scalability.  The initial Stellar implementation uses
two-phase blocks, but the \SPEEDEX{} phase is still implemented
sequentially.  As a result, the initial implementation is simple (adding
only $\sim$5,000 lines to the server daemon) and the primary
benefits are economic.  However, because the transaction semantics are commutative,
engineers can work to parallelize the implementation as needed,
without formally upgrading the protocol (which is more difficult than
releasing a software update).

\subsection{Caches and \Tat{}}
\label{sec:tatimpldetails}

\Tat{} spends most of its runtime computing demand queries.
 Each query consists of several binary searches over large lists, so
the runtime depends heavily on memory latency and cache performance.  
Towards the end of \Tat{}, when the algorithm takes small steps, one query reads almost exactly the same memory locations as the previous query, so the cache miss rate can be extremely low.

Instead of querying the offer tries directly, we precompute for each asset pair a list that records, for each unique limit price, the amount of an asset offered for sale below the price (\S \ref{sec:tatpreprocessing}).  Laying out this information contiguously improves cache performance.

We also execute the binary searches of one \Tat{} iteration in parallel.
One primary thread computes price updates and wakes helper threads.  However, each round of \Tat{} is already fast on one thread---with $50$ assets and millions of offers, one round takes $400$--$600\mu s$.
To minimize synchronization latency and avoid letting the kernel migrate threads between cores (which harms cache performance), we operate these helper threads via spinlocks and memory fences.
In the tests of \S\ref{sec:tateval}, we see minimal benefit beyond 4--6 helper threads, but this suffices to reduce each query to $50$--$150\mu s$. 

Finally, there is a tradeoff between running more copies of \Tat{} with different settings and the performance of each copy.    More concurrent replicas of \Tat\ mean more cache traffic and higher cache miss rates.

We accelerate the rest of \Tat{} by exclusively using fixed-point arithmetic (rather than floating-point).



\input{trie_design}


%% file: trie_design.tex
\subsection{Batched Trie Design}

Our tries use a fan-out
of 16 and hash nodes with the 32-byte BLAKE2b cryptographic hash \cite{blake2}.
Both the layout of trie nodes and the work partitioning are designed
to avoid having multiple threads writing to the same cache line.

The commutativity of \SPEEDEX{}'s semantics opens
up an efficient design space for our data structures, which need only
materialize 
state changes once per block.  
Tries need only recompute
a root hash once per block, for example, instead of after every modification.
Threads locally build tries recording insertions,
which are merged together in one batch operation (which is also parallelizable
by redividing local tries into disjoint key ranges).
Deletions (when offers are cancelled)
are implemented via atomic flags
on trie nodes; to enable efficient cleanup of deleted nodes, each node stores the number of deleted nodes beneath it.
To facilitate efficient work distribution, each node also stores the number of leaves below it. 

\SPEEDEX{} builds in every block
an ephemeral trie that logs which accounts are modified;
specifically, it maps an account ID to a list of its transactions and
to the IDs of transactions from other accounts that modified it.
This enables
construction of short proofs of account state changes.
This trie also uses the same key space as the main
account state trie, which lets \SPEEDEX{} use the ephemeral trie to 
efficiently divide work on the (much larger) account trie.

Memory allocation for an emphemeral trie is trivial because no ephemeral trie node
is carried over from one block to the next.
Every thread has a local
arena, allocation simply increments an arena index, and garbage collection
means just setting the index to 0 at the end of a block.
We find it to be not a problem if some of the memory in the arena is wasted;
we allocate the potential children of an ephemeral trie node contiguously,
so a node need only store a 4-byte base pointer (buffer index) and a bitmap
denoting the active children.  This lets each ephemeral trie node fit in one
64-byte cache line.

\XXX{I don't really have any evidence for this, and the effect is marginal if 
each account sends only a small number of txs (and the account db is not sorted 
by account number -- it used to be, in early experiments for the SOSP paper, before I realized
that that would bias results).

Building this
trie implicitly sorts transactions by account ID.  Block proposals that are sorted
are slightly faster to execute, as there is less cache contention between threads working
on contiguous segments of sorted transaction list (this is not checked by nodes,
so a faulty block producer could slightly reduce performance by issuing unsorted but otherwise valid proposals).
}


%% file: related_work.tex
\section{Related Work}
\label{sec:relwork}
\paragraph{Blockchain Scaling.}
Our approach is inspired by Clements et al.~\cite{clements2015scalable}, who improve performance in the Linux kernel through commutative syscall semantics.

Chen et al.~\cite{chen2021forerunner} speculatively execute Ethereum transactions to achieve a $\sim$6x overall execution speedup.
Other approaches to concurrent execution include optimistic concurrency control~\cite{gelashvili2022block,sealevel},
invalidating conflicting transactions~\cite{hyperledgerfabric}, broadcasting
conflict resolution information~\cite{dickerson2019adding,anjana2019efficient}, or partitioning transactions into nonconflicting sets~\cite{bartoletti2020true,yu2018parallel,stellartxfootprintcap}.
This problem is related to that of building deterministic databases and software transactional memory~\cite{prasaad2020handling,xia2019litm,thomson2012calvin}.  Li et al.~\cite{li2012making} build a distributed database where some transactions are tagged as commutative.

Empirical work~\cite{saraph2019empirical,garamvolgyi2022utilizing} finds that a small number of Ethereum contracts, often token contracts,
are historically responsible for the majority of conflicts that limit optimistic execution.  
A recent Solana~\cite{yakovenko2018solana} outage
resulted in part when many transactions conflicted on one orderbook contract~\cite{solanaoutage}.


Project Hamilton~\cite{opencbdc} develops a CBDC payments platform.  The authors find that totally-ordered semantics become a performance bottleneck.  Unlike \SPEEDEX{}, which stores asset balances in accounts,
this system requires the more restrictive unspent transaction output (UTXO) model.

Some systems move transaction execution off-chain, into so-called ``Layer-2'' networks,
each with different capabilities, perfomance, interoperability, and security tradeoffs
\cite{poon2016bitcoin,kalodner2018arbitrum,polygon,optimisticrollup,poon2017plasma,zkrollup,loopringdesigndoc}.
Other blockchains
\cite{ethsharding,skychain,polkadot,hyperledgerfabric,nearsharding}
split state into concurrently-running shards,
at the cost of complicating cross-shard transactions.

\paragraph{(Distributed) Exchanges.}
Budish et al.~\cite{budish2015high,budish2014implementation} argue that exchanges should process orders in batches to combat automated arbitrage and improve liquidity. 

Other defenses against front-running include cryptographic commit-reveal schemes
\cite{clineclockwork,zhang2022flash,schmid2021secure,goes2021anoma} or ``fair'' ordering schemes
that assume a bounded fraction of malicious nodes \cite{kelkar2020order,zhang2020byzantine,chainlink}.
The front-running attacks that \SPEEDEX{} prevents are not guaranteed to be blocked in these schemes.  For example,
a replica might plausibly front-run a transaction in~\cite{kelkar2020order} by investing in lower-latency network links between itself
and other replicas than other replicas have with each other, and commit-reveal schemes do not prevent statistical front-running (guessing
the contents of a transaction).  
\XXX{@ashish what do you think}

Some blockchains build limit-order DEX mechanisms natively~\cite{stellar,binancedexoldintro} or as smart contracts~\cite{serum}.
Smart contracts known as Automated Market-Makers
(AMMs)~\cite{uniswapv2,bancor,egorov2019stableswap,balancer} facilitate passive market-making on-chain~\cite{angeris2019analysis}.

0x and a past version of Loopring~\cite{warren20170x,loopring} allow settlement on-chain of orders matched off-chain, in pairs or in cycles.
StarkEx~\cite{starkware,ben2018scalable} gives cryptographic tools to prove correctness of an off-chain exchange.

CoWSwap \cite{dfusion,cowswapproblem} uses mixed-integer programming
to clear offers in batches of at most 100 \cite{dfusionbatch}. Solvers compete to produce the best solution. The former Binance DEX \cite{binancedexold} computed per-asset-pair prices in each block.  The Penumbra DEX
uses homomorphic encryption to privately make batch swaps against an AMM, but cannot let users set limit prices \cite{penumbraswap}.

\paragraph{Price Computation.}
Our algorithms solve instances of the special case of the Arrow-Debreu exchange market \cite{arrow1954existence} where every utility function is linear.
Equilibria can be approximated in these markets using combinatorial algorithms such as those of Jain et al.~\cite{jain2003approximating} and Devanur et al.~\cite{devanur2003improved} and exactly via the ellipsoid method and simultaneous diophantine approximation \cite{jain2007polynomial}.  Duan et al.~\cite{duan2015combinatorial} construct an exact combinatorial algorithm, which Garg et al.~\cite{garg2019strongly} extend to an algorithm with strongly-polynomial running time.
Ye~\cite{ye2008path} gives a path-following interior point method, and Devanur et al.~\cite{devanur2016rational} construct a convex program.
Codenotti et al.~\cite{codenotti2005polynomial,codenotti2005market} show that a version of the \Tat{} process \cite{arrow1959stability2}
 converges to an approximate equilibrium in polynomial time.
Garg et al.~\cite{garg2019auction} give another algorithm based on demand queries.

%% file: conclusion.tex
\section{Conclusion}

\SPEEDEX{} is a fully on-chain DEX that can scale to more than 200,000 transactions per second with tens of millions of open trade offers.
\SPEEDEX{} requires no offchain rollups and no sharding of the exchange's logical state.  
To make \SPEEDEX{} faster, one can simply give \SPEEDEX{} more CPU cores, without changing the semantics or user interface.  
Because \SPEEDEX{} operates as a logically-unified platform, instead of a sharded network, \SPEEDEX{} does not fragment liquidity between subsystems and creates no cross-rollup arbitrage.

In addition, \SPEEDEX{} displays several independently useful economic properties.  
It eliminates risk-free front running; any user who can get their offer to the exchange before a block cutoff time can get the same exchange rate as every other trader.
\SPEEDEX{} also eliminates internal arbitrage, which disincentivizes network spam.
And finally, \SPEEDEX{} eliminates the need to transact through
intermediate, reserve currencies, 
instead allowing a user to trade directly from one asset to any other asset listed on the exchange, 
with the same or better market liquidity as the trader would have gotten by trading through a series of intermediate currencies.

\SPEEDEX{} is free software, available at \speedexgithub{}.


%% file: acks.tex
\section*{Acknowledgements}


This research was supported by the 
Stanford Future of Digital Currency Initiative, 
the Stanford Center for Blockchain Research,
the Office of Naval Research (ONR N00014-19-1-2268),
and the Army Research Office (76412CSII).
The Stellar blockchain integration was funded by
and performed at the Stellar Development Foundation.

The authors wish to thank the anonymous reviewers
and our shepherd Siddhartha Sen
for their valuable feedback,
and
thank CloudLab~\cite{cloudlab}
for providing resources for our experiments.

%% file: arrow_debreu_mapping.tex
\section{Mathematical Model Underlying \SPEEDEX{}}
\label{sec:speedexandad}

Mathematically, \SPEEDEX{} relies on a correspondence between a batch of trade offers
and an instance of a linear Arrow-Debreu Exchange Market \cite{arrow1954existence}.  Specifically, \SPEEDEX{}'s batch computation
is equivalent to the problem of computing equilibria in these markets.

\subsection{Arrow-Debreu Exchange Markets}
\label{sec:arrowdebreu}

The Arrow-Debreu Exchange Market is a classic model from the economics and theoretical computer science literature.
Conceptually, there exists in this market a set of independent \textit{agents}, each with its own \textit{endowment} of goods.  Each agent
has some set of preferences over possible collections of goods.  These goods are tradeable on an open market, and
agents, all at the same time, make any set of trades that they wish
with \textit{the market} (or \emph{auctioneer}), not directly with each other.

\begin{definition}[Arrow-Debreu Exchange Market]
An Arrow-Debreu Exchange Market consists of a set of goods $\assetset$ and a set of agents $j\in \lbrace 1,...,M\rbrace $.  
Every agent $j$ has a utility function $u_j(\cdot)$ and an endowment $e_j\in \mathbb{R}^{\assetsize}_{\geq 0}$.  

When the market trades at prices $p\in\mathbb{R}^{\assetsize}_{\geq 0}$,
 every agent sells their endowment to the market in exchange for revenue $s_j=p\cdot e_j$, 
 which the agent immediately spends at the market to buy back an optimal bundle of goods $x_j\in \mathbb{R}^{\assetsize}_{\geq 0}$ - that is, 
$x_j=\argmax_{x:\sum{\asset{A}\in\assetset} x_{\asset{A}}p_{\asset{A}}\leq s_j} u_j(x)$.
\end{definition}

There are countless variants on this definition.  Typically the utility functions are assumed to be quasi-convex.  Some variants
include stock dividents,
corporations, production of new goods from existing goods, and multiple trading rounds.  
\SPEEDEX{} uses only the model outlined above---\SPEEDEX{} looks only
at snapshots of the market, i.e., once per block, and computes batch results for each block independently.

One potential objection to the above definition is that it assumes that the 
abstract market has sufficient quantities available so that every agent can
make its preferred trades.  We say that a market is at \textit{equilibrium} when agents can make their preferred trades and the market does not have
a deficit in any good.

\begin{definition}[Market Equilibrium]
	An equilibrium of an Arrow-Debreu market is a set of prices $p$ and an allocation $x_j$ for every agent $j$, such that for all goods $\asset{A}$, 
	$\sum_j e_{\asset{A},j} \geq \sum_j x_{\asset{A},j}$, and $x_j$ is an optimal bundle for agent $j$.  The inequality for asset $\asset{A}$ is tight whenever $p_{\asset{A}}$ is nonzero.
\end{definition}

Note that an equilibrium includes both a set of market prices and a choice of a utility-maximizing set of goods for each agent.
Say, for example, there are two goods $\asset{A}$ and $\asset{B}$, and one unit of each is sold by other agents to the market.  
If two agents are indifferent
to receiving either good, then the equilibrium must specify whether the first receives $\asset{A}$ or $\asset{B}$, 
and vice versa for the second.  It would not be
a market equilibrium for both of these agents to purchase a unit of $\asset{A}$ and no units of $\asset{B}$.

\subsection{From \SPEEDEX{} to Exchange Markets}

\SPEEDEX{} users do not submit abstract utility functions to an abstract market.  However, most natural types of trade offers
can be encoded as a simple utility function.

Specifically, our implementation of \SPEEDEX{} accepts limit sell orders of the following form.

\begin{definition}[Limit Sell Offer]
A {\it Sell Offer} ($\asset{S}$, $\asset{B}$, $e$, $\alpha$) is request to sell $e$ units of good $\asset{S}$
 in exchange for some number $k$ units of good $\asset{B}$, subject to the condition that $k\geq \alpha e$.
\end{definition}

The user who submits this offer implicitly says that they value $k$ units of $\asset{B}$ more than $e$ units of $\asset{S}$ 
if and only if $k\geq \alpha e$.
These preferences are representable as a linear utility function.

\begin{theorem}
\label{thm:linearutils}

Suppose a user submits a sell offer ($\asset{S}$, $\asset{B}$, $e$, $\alpha$).  The optimal behavior of this offer (and the user's implicit preferences)
is equivalent to maximizing the function $u(x_{\asset{S}}, x_{\asset{B}}) = \alpha x_{\asset{B}} + x_{\asset{B}}$ (for $x_{\asset{S}}, x_{\asset{B}}$ amounts of goods $\asset{S}$ and $\asset{B}$).

\end{theorem}

\begin{proof}

Such an offer makes no trades if $p_{\asset{S}}/p_{\asset{B}} <\alpha$ and trades in full if $p_{\asset{S}}/p_{\asset{B}}>\alpha$.

The user starts with $k$ units of $\asset{S}$.  In the exchange market model,
the user can trade these $k$ units of $\asset{S}$ in exchange for any quantities $x_{\asset{S}}$ of $\asset{S}$ 
and $x_{\asset{B}}$ of ${\asset{B}}$,
subject to the constraint that $p_{\asset{S}}x_{\asset{S}} + p_{\asset{B}}x_{\asset{B}} \leq kp_{\asset{S}}$.

The function $u(x_{\asset{S}}, x_{\asset{B}}) = \alpha x_{\asset{B}} + x_{\asset{S}}$ is maximized, subject to the above constraint,
 by $(x_B, x_{\asset{S}})=(0, k)$ precisely when $p_{\asset{S}}/p_B <\alpha$ and by 
 $(x_{\asset{B}}, x_{\asset{S}})=(kp_{\asset{S}}/p_{\asset{B}}, 0)$ otherwise
 (and by any convex combination of the two when $p_{\asset{S}}/p_{\asset{B}} = \alpha$).  These allocations correspond exactly
 to the optimal behavior of a limit sell offer.
 \end{proof}

Note that these utility functions have nonzero marginal utility for only two types of assets, and are not arbitrary linear utilities.
Ramseyer et al.~\cite{ramseyer2022batch} find anecdotal evidence that this subclass of utility functions may be analytically more tractable than
the case of general linear utilities.

\subsection{Existence of Unique* Equilibrium Prices}
\label{sec:uniquevalsexist}

\begin{theorem}
All of the market instances which \SPEEDEX{} considers contain an equilibrium with nonzero prices.

\end{theorem}

\begin{proof}

All of the utilities of agents derived from limit sell offers are linear (Theorem \ref{thm:linearutils}), and have a nonzero
marginal utility on the good being sold.

This means our market instances trivially satisfy condition (*) of
Devanur et al.~\cite{devanur2016rational}.  Existence
of an equilibrium with nonzero prices follows therefore from Theorem 1 of \cite{devanur2016rational}.
\end{proof}

In fact, all of the equilibria in a market instance contain the same equilibrium prices, unless there are two sets of assets across
which no trading activity occurs.  In such a case, one might be able to uniformly increase or decrease all the prices together on one set of assets, relative
to the other set of assets.

\begin{theorem}
\label{thm:partitioning}
Suppose there are two equilibria $(p,x)$ and $(p^\prime, x^\prime)$ and there exist two assets $\asset{A}$ and $\asset{B}$ for which 
$p_{\asset{A}}/p_{\asset{B}} < p_{\asset{A}}^\prime/p_{\asset{B}}^\prime$.

Then it must be the case that there is a partitioning of the assets $\assetset_1, \assetset_2$ with $A\in\assetset_1,B\in\assetset_2$ 
such that both equilibria include no trading activity across the partition.
\end{theorem}

\begin{proof}

Consider the set of offers trading from $\asset{A}$ to $\asset{B}$.  Let $Z_{\asset{A},\asset{B}}(r)$ be the set of amounts of asset $\asset{A}$
that may be sold (when every agent receives an optimal bundle) by these offers to the market at an exchange rate $r=p_{\asset{A}}/p_{\asset{B}}$.
Observe that if $r_1 < r_2$, then every $z_1\in Z_{\asset{A},\asset{B}}(r_1)$ is no more than than any $z_2\in Z_{\asset{A},\asset{B}}(r_2)$ (as sell offers
always prefer higher exchange rates).

At the equilibrium $(p,x)$, let $z_{\asset{A}, \asset{B}}$ be the total amount of $\asset{A}$ sold for $\asset{B}$ for every asset pair
(and $z^\prime_{\asset{A}, \asset{B}}$ similarly for $(p^\prime,
x^\prime)$).  
Note that $z_{\asset{A}, \asset{B}} \in Z_{\asset{A},\asset{B}}(p_{\asset{A}}/p_{\asset{B}})$.

Suppose that there exists a pair of assets $\asset{A},\asset{B}$ as in the theorem statement.  Then there exists a set of assets $\assetset_1$
such that for every asset pair $\asset{C}\in\assetset_1$ and $\asset{D}\notin \assetset_1$, 
$p_{\asset{C}}/p_{\asset{D}} < p_{\asset{C}}^\prime/p_{\asset{D}}^\prime$.

For each of these asset pairs, we must have that $z_{\asset{C},\asset{D}} \leq z^\prime_{\asset{C},\asset{D}}$,
$z_{\asset{D},\asset{C}} \geq z^\prime_{\asset{D},\asset{C}}$,
and $\frac{p_{\asset{C}}}{p_{\asset{D}}} z_{\asset{C},\asset{D}} \leq \frac{p^\prime_{\asset{C}}}{p^\prime_{\asset{D}}} z^\prime_{\asset{C},\asset{D}}$.
Combining these equations gives
\begin{equation*}
p_{\asset{C}}z_{\asset{C},\asset{D}} - p_{\asset{D}}z_{\asset{D},\asset{C}} 
\leq (p_{\asset{C}}^\prime z^\prime_{\asset{C},\asset{D}} - p_{\asset{D}}^\prime z_{\asset{D},\asset{C}}^\prime) p_{\asset{D}}/p_{\asset{D}}^\prime
\end{equation*}
Each of these inequalities is tight if and only if $z_{\asset{C},\asset{D}} = 0$.

It is without loss of generality to rescale $p^\prime$ so that $p_{\asset{D}}/p_{\asset{D}}^\prime < 1$ for all $\asset{D}\notin\assetset_1$.
Thus,
\begin{equation*}
p_{\asset{C}}z_{\asset{C},\asset{D}} - p_{\asset{D}}z_{\asset{D},\asset{C}} 
\leq (p_{\asset{C}}^\prime z^\prime_{\asset{C},\asset{D}} -
p_{\asset{D}}^\prime z_{\asset{D},\asset{C}}^\prime)
\end{equation*}
Because $(p,x)$ and $(p^\prime, x^\prime)$ are equilibria, we must have that
\begin{dmath*}
0=\sum_{\asset{C}\in\assetset_1} \sum_{\asset{D}\notin\assetset_1} p_{\asset{C}}z_{\asset{C},\asset{D}} - p_{\asset{D}}z_{\asset{D},\asset{C}}\\
\leq \sum_{\asset{C}\in\assetset_1} \sum_{\asset{D}\notin\assetset_1} 
p^\prime_{\asset{C}}z^\prime_{\asset{C},\asset{D}} - p^\prime_{\asset{D}}z^\prime_{\asset{D},\asset{C}}
\end{dmath*}
But the second inequality is tight only if each $z_{\asset{C},\asset{D}}=0$.

Hence, $(p^\prime, x^\prime)$ can only be an equilibrium if there exists a partitioning of the assets that separates $\asset{A}$ and $\asset{B}$,
and for which there is no trading activity between the sets in either equilibrium.
\end{proof}

\begin{corollary}
\label{corollary:connected}

Let $(p,x)$ be an equilibrium.

Construct an undirected graph $G=(V,E)$ with one vertex for each asset,
and an edge $e=(\asset{A},\asset{B})\in E$ if, at equilibrium, any $\asset{A}$ is sold for $\asset{B}$ or any $\asset{B}$ is sold for $\asset{A}$.

If $G$ is connected, then the market equilibrium prices $p$ are unique (up to uniform rescaling).

\end{corollary}

\begin{proof}

If the theorem hypothesis holds, then for any other equilibrium $(p^\prime, x^\prime)$, it must be the case that
for every asset pair $(\asset{A},\asset{B})$, $p_{\asset{A}}/p_{\asset{B}}=p_{\asset{A}}^\prime/p_{\asset{B}}^\prime$. 
By Theorem \ref{thm:partitioning}, if this did not hold,
then there would exist a partitioning of $V$ into two sets of assets,
across which there is no trading at equilibrium $(p,x)$ (contradicting the assumption that $G$ is connected).
\end{proof}

%% file: approximation.tex
\section{Approximation Error}
\label{apx:error}

\SPEEDEX{} measures two forms of approximation error:  first, every trade is charged a $\varepsilon$ transaction
commision, and second, some offers with in-the-money limit prices might not be able to be executed (while preserving asset conservation).
Formally, the output of the batch price computation is a price $p_{\asset{A}}$ on each asset $\asset{A}$, and a trade amount $x_{\asset{A},\asset{B}}$ denoting the amount of $\asset{A}$
sold in exchange for $\asset{B}$.  

Formally, we say that the result of a batch price computation is $(\varepsilon, \mu)$-approximate if:
\begin{enumerate}
	\item[1] 

		Asset conservation is preserved with an $\varepsilon$ commission.  The amount of ${\asset{A}}$ sold to the auctioneer,
		$\Sigma_{\asset{B}} x_{\asset{A},\asset{B}}$, must exceed the amount of ${\asset{A}}$ bought from the auctioneer,
		$\Sigma_{\asset{B}} (1-\varepsilon)\frac{p_{\asset{B}}}{p_{\asset{A}}}x_{\asset{B},\asset{A}}$.

	\item[2]
		No offer trades outside of its limit price.  That is to say, an offer selling ${\asset{A}}$ for ${\asset{B}}$ with a limit price of $r$
		cannot execute if $\frac{p_{\asset{A}}}{p_{\asset{B}}} < r$.

	\item[3]
		No offer with a limit price ``far'' from the batch exchange rate does not trade.
		That is to say, an offer selling ${\asset{A}}$ for ${\asset{B}}$ with a limit price of $r$ must trade in full if
		$r < (1-\mu)\frac{p_{\asset{A}}}{p_{\asset{B}}}$.

		Intuitively, the lower the limit price, the more an offer prefers trading to not trading.

\end{enumerate}

This notion of approximation is closely related to but not exactly the same as 
notions of approximation used in the theoretical literature on 
Arrow-Debreu exchange markets (e.g.,~\cite{codenotti2005market}, Definition 1).
In particular, we find it valuable in \SPEEDEX{} to distinguish between the two types of approximation error
(and measure each separately) and \SPEEDEX{} must maintain certain guarantees exactly
(e.g., assets must be conserved,
and no offer can trade outside its limit price).

%% file: tatonnement_mods.tex
\section{\Tat{} Modifications}
\label{apx:tatmods}

\subsection{Price Update Rule}
\label{sec:priceupdate}

\let\realSigma\Sigma
\let\Sigma\sum

One significant algorithmic difference between the \Tat{} implemented within \SPEEDEX{}
and the \Tat{} described in Codenotti et al.~\cite{codenotti2005market} is the method in which \Tat{} adjusts 
prices in response to a demand query.  Codenotti et al.\ use an additive rule that they find amenable to
theoretical analysis.  If $Z(p)$ is the market demand at prices $p$, they update prices according to the following rule:
\begin{equation}
p_{\asset{A}} \leftarrow p_{\asset{A}}  + Z_{\asset{A}} (p)\delta
\end{equation}
for some constant $\delta$.  The authors show that there is a sufficiently small $\delta$ so that \Tat{} is guaranteed
to move closer to an equilibrium after each step.

The relevant constant is unfortunately far too small to be usable in practice, and more generally, we want an algorithm
that can quickly adapt to a wide variety of market conditions (not one that always proceeds at a slow pace).

First, we update prices multiplicatively, rather than additively.  This dramatically reduces the number of required rounds, especially
when \Tat{} starts at prices that are far from the clearing prices.
\begin{equation}
p_{\asset{A}}  \leftarrow p_{\asset{A}}  (1 + Z_{\asset{A}} (p)\delta)
\end{equation}

Second, we normalize asset amounts by asset prices, so that our algorithm will be invariant to redenominating an asset.
It is equivalent to trade 100 pennies or 1 USD, and our algorithm performs better when it can take that kind of context
into account.
\begin{equation}
p_{\asset{A}}  \leftarrow p_{\asset{A}}  (1 + p_{\asset{A}} Z_{\asset{A}} (p)\delta)
\end{equation}

Next, we make $\delta$ a variable factor.  
We use a heuristic 
to guide the dynamic adjustment.  
Our experiments used the $l^2$ norm of the
price-normalized demand vector, $\Sigma_{\asset{A}} (p_{\asset{A}} Z_{\asset{A}} (p))^2$; other natural
heuristics (i.e. other $l^p$ norms) perform comparably (albeit not quite as well).
In every round, \Tat{} computes this heuristic at its current set of candidate prices, and
at the prices to which it would move should it take a step with the current step size.  
If the heuristic goes
down, \Tat{} makes the step and increases the step size, and otherwise decreases the step size.
This is akin to a backtracking line search \cite{armijo1966minimization,boyd2004convex} with a weakened termination condition.
\begin{equation}
p_{\asset{A}}  \leftarrow p_{\asset{A}}  (1 + p_{\asset{A}} Z_{\asset{A}} (p)\delta_t)
\end{equation}

Finally, we normalize adjustments by a trade volume factor $\nu_{\asset{A}}$.
Without this adjustment factor, computing prices when one asset is traded much less
than another asset takes a large number of rounds, simply because the lesser traded asset's price updates are always
of a lower magnitude than those of the more traded asset.  
Many other numerical optimization problems run most quickly
when gradients are normalized (e.g., see~\cite{benzi2002preconditioning}).

$\nu_{\asset{A}} $ need not be perfectly accurate---indeed, knowing
the factor exactly would require first computing clearing
prices---but we can estimate it well enough from 
the trading volume in prior blocks and from trading volume
in earlier rounds of \Tat{} (specifically, we use the
minimum of the amount of an asset sold to the auctioneer and the amount bought from the auctioneer).
Real-world deployments could estimate these factors using external market data.

Putting everything together gives the following update rule:
\begin{equation}
p_{\asset{A}} \leftarrow p_{\asset{A}} \left(1+p_{\asset{A}}
Z_{\asset{A}} (p)\delta_t\nu_{\asset{A}}\right)
\end{equation}
The step size is represented internally as a 64-bit integer and a
constant scaling factor.  As mentioned in \S \ref{sec:multiinstance}, 
we run several copies of \Tat{} in parallel with
different scaling factors and different volume normalization
strategies and take whichever finishes first as the result.

\subsubsection{Heuristic Choice}

A natural question is why do we use the seemingly theoretically unfounded
$l^2$ norm of the demand vector as our line-search heuristic.
A typical line search in an optimization context uses
the convex objective function of the optimization problem (e.g., \cite{boyd2004convex}).
Devanur et al.~\cite{devanur2016rational} even give a convex objective function
for computing exchange market equilibria, which we reproduce below (in a simplified
form):

\begin{equation}
\Sigma_{i:mp_i < \frac{p_{\asset{S}_i}}{p_{\asset{B}_i}}}
p_{\asset{S}_i}E_i\ln\left(mp_i
\frac{p_{\asset{S}_i}}{p_{\asset{B}_i}}\right) - y_i\ln (mp_i)
\end{equation}
for $mp_i$ the minimum limit price of an offer $i$ that sells $E_i$ units of good $\asset{S}_i$ and buys good $\asset{B}_i$,
and $y_i=x_ip_{\asset{S}_i}$ for $x_i$ the amount of $\asset{S}_i$ sold by the offer to the market.

This objective is accompanied by an asset conservation constraint for each asset $\asset{A}$:
\begin{equation}
\Sigma_{i:\asset{S}_i=\asset{A}} y_i = \Sigma_{i:\asset{B}_i=\asset{A}} y_i
\end{equation}
However, unlike the problem formulation in \cite{devanur2016rational}, \Tat{} does not have decision
variables $\lbrace y_i\rbrace$. Rather, \Tat{} pretends offers respond rationally
to market prices, and then adjusts prices so that constraints become satisfied.
As such, mapping our algorithms onto the above formulation would mean that $y_i=p_{\asset{S}_i}E_i$
if $mp_i < \frac{p_{\asset{S}_i}}{p_{\asset{B}_i}}$ and $0$ otherwise (although \S \ref{apx:demandsmoothing} would slightly change
this picture).
This would make the objective universally $0$, and thus not useful.

We could incorporate the constraints into the objective by using the Lagrangian of the above problem,
which gives the objective
\begin{equation}
\label{eqn:l1}
\Sigma_{\asset{A}}\lambda_{\asset{A}}(\Sigma_{i:\asset{S}_i=\asset{A}} y_i(p) - \Sigma_{i:\asset{B}_i=\asset{A}} y_i(p))
\end{equation}
for a set of langrange multipliers $\lbrace \lambda_{\asset{A}}\rbrace$.  

We write $y_i(p)$ to
denote that in this formulation, offer behavior is directly a function of prices.  It appears
difficult to use equation \ref{eqn:l1} directly as an objective to minimize, as
it is nonconvex and the gradients of the functions $y_i(\cdot)$ are numerically unstable
(even with the application of \S \ref{apx:demandsmoothing}).

However, observe that equation \ref{eqn:l1} is another way of writing
``the $l^1$ norm of the net demand vector'' (weighted
by the lagrange multipliers).  
We use the $l^2$ norm instead of the $l^1$ to sidestep the need to actually solve for these multipliers.

An observant reader might notice that the derivative of Equation \ref{eqn:l1} with respect to $\lambda_{\asset{A}}$ is
the amount by which (the additive version of) \Tat{} updates $p_{\asset{A}}$.  This might suggest using $p_{\asset{A}}$ in place of $\lambda_{\asset{A}}$ in
equation \ref{eqn:l1}.  However, that search heuristic performs extremely poorly.

\subsection{Demand Smoothing}
\label{apx:demandsmoothing}

Observe that the demand of a single offer is a (discontinuous) step function;
an offer trades in full when the market exchange rate exceeds its limit price,
and not at all when the market rate is less than its limit price.

These discontinuities are difficult for \Tat. (Analogously, many optimization
problems struggle on nondifferentiable objective functions.)  As such, 
we approximate the behavior of each offer with a continuous function.

Recall that \S \ref{apx:error} measures one form of approximation error
(using the parameter $\mu$) which asks how closely realized offer behavior
matches optimal offer behavior.  Specifically, \SPEEDEX{} wants to maintain
the guarantee that every offer (selling $\asset{A}$ for $\asset{B}$) with a limit price
below $(1-\mu) \frac{p_{\asset{A}}}{p_{\asset{B}}}$ trades in full, and those with limit prices
above $\frac{p_{\asset{A}}}{p_{\asset{B}}}$ trade not at all.

As such, \SPEEDEX{} has the flexibility to specify offer behavior
on the gap between $(1-\mu)\frac{p_{\asset{A}}}{p_{\asset{B}}}$ and $\frac{p_{\asset{A}}}{p_{\asset{B}}}$.
Instead of a step function, \SPEEDEX{} linearly interpolates across the gap.
That is to say, if $\alpha=\frac{p_{\asset{A}}}{p_{\asset{B}}}$,
we say that an offer with limit price $(1-\mu)\alpha \leq \beta\leq \alpha$ sells
an $\frac{\alpha-\beta}{\mu\alpha}$ fraction of its assets.

Observe that as $\mu$ gets increasingly small, this linear interpolation becomes
an increasingly close approximation of a step function.  This explains some of the behavior
in Figure \ref{fig:mintxs}, particularly why the price computation problem gets increasingly difficult
as $\mu$ decreases.

\subsection{Periodic Feasibility Queries}

\Tat{}'s linear interpolation simplifies computing each round, but also
restricts the range of prices that meet the approximation criteria, as
it does not capitalize on the flexibility we have in handling offers
within $\mu$ of the market price.  As a result, \Tat{} may arrive at
adequate prices without recognizing that fact.
To identify good valuations, \SPEEDEX{}
runs the more expensive linear program every 1,000 iterations of \Tat{}.

%% file: linear_program.tex
\section{Linear Program}
\label{apx:lp}

Recall that the role of the linear program in \SPEEDEX{} is to compute the maximum amount
of trading activity possible at a given set of prices.  That is to say, \Tat{} first computes
an approximate set of market clearing prices, and then \SPEEDEX{} runs this linear program 
taking the output of \Tat{} as a set of input, constant parameters.

Throughout the following, we denote the price of an asset $\asset{A}$ (as output from \Tat{}) as $p_{\asset{A}}$,
and the amount of $\asset{A}$ sold in exchange for $\asset{B}$ as $x_{\asset{A},\asset{B}}$.  We will also denote the two forms of
approximation error as $\varepsilon$ and $\mu$, as defined in \S \ref{apx:error}.

To maintain asset conservation, the linear program must satisfy the following constraint for every asset $\asset{A}$:

\begin{equation*}
\Sigma_{\asset{B}} x_{\asset{A},\asset{B}} \geq \Sigma_{\asset{B}}
(1-\varepsilon)\frac{p_{\asset{B}}}{p_{\asset{A}}}x_{\asset{B},\asset{A}}
\end{equation*}

Define $U_{\asset{A},\asset{B}}$ to be the upper bound on the amount of $\asset{A}$ that is available for sale
by all offers with in the money limit prices (i.e., limit prices at or
below $\frac{p_{\asset{A}}}{p_{\asset{B}}}$),
and define $L_{\asset{A},\asset{B}}$ to be the lower bound on the amount of $\asset{A}$ that must be exchanged for $\asset{B}$
if \SPEEDEX{} is to be $\mu$-approximate (i.e., execute all offers
with minimum prices at or below
$(1-\mu)\frac{p_{\asset{A}}}{p_{\asset{B}}}$, as described in \S \ref{apx:error}).

Then the linear program must also satisfy the constraint, for every asset pair $(\asset{A},\asset{B})$,

\begin{equation*}
L_{\asset{A},\asset{B}}\leq x_{\asset{A},\asset{B}} \leq U_{\asset{A},\asset{B}}
\end{equation*}

Informally, the goal of our linear program is to maximize the total amount of trading activity.
Any measurement of trading activity needs to be invariant to redenominating assets; intuitively,
it is the same to trade 1 USD or 100 pennies.  As such, the objective of our linear program is:

\begin{equation*}
\Sigma_{\asset{A},\asset{B}} p_{\asset{A}} x_{\asset{A},\asset{B}}
\end{equation*}

Putting this all together gives the following linear program (let $\assetset$ be the set of all assets):

\begin{align}
\max ~& \Sigma_{\asset{A},\asset{B}} p_{\asset{A}} x_{\asset{A},\asset{B}}\\
s.t. ~& p_{\asset{A}} L_{\asset{A},\asset{B}}\leq p_{\asset{A}} x_{\asset{A},\asset{B}}
\leq p_{\asset{A}} U_{\asset{A},\asset{B}}(p)~~~\forall \asset{A},\asset{B}\in\assetset,~(\asset{A}\neq \asset{B})\\
~& p_{\asset{A}}\Sigma_{\asset{B}\in \assetset}x_{\asset{A},\asset{B}} 
\geq (1-\varepsilon)\Sigma_{\asset{B}\in\assetset}p_{\asset{B}} x_{\asset{B},\asset{A}}~~~ \forall \asset{A}\in\assetset
\end{align}

From the point of view of the linear program, $p_{\asset{A}}$ is a constant (for each asset $\asset{A}$).
As such, this optimization problem is in fact a linear program.

It is possible that \Tat{} could output prices where this linear program is infeasible (this is the
case of the \Tat{} timeout, as discussed in \S \ref{sec:tateval}).  In these cases,
we set the lower bound on each $x_{\asset{A},\asset{B}}$ to be $0$ instead of $L_{\asset{A},\asset{B}}$.  This change makes the program
always feasible (e.g., an assigment of each variable to $0$ satisfies the constraints).

Observe that as written, every instance of the variable $x_{\asset{A},\asset{B}}$ appears adjacent to $p_{\asset{A}}$.
We can simplify the program by replacing each occurrence of $p_{\asset{A}} x_{\asset{A},\asset{B}}$ by a new variable $y_{\asset{A},\asset{B}}$.
After solving the program, we can compute $x_{\asset{A},\asset{B}}$ as $\frac{y_{\asset{A},\asset{B}}}{p_{\asset{A}}}$.

This substitution gives the following linear program:

\begin{align}
\max ~& \Sigma_{\asset{A},\asset{B}} y_{\asset{A},\asset{B}}\\
s.t. ~& p_{\asset{A}} L_{\asset{A},\asset{B}}\leq y_{\asset{A},\asset{B}}\leq p_{\asset{A}} U_{\asset{A},\asset{B}}(p)~~~\forall (\asset{A},\asset{B}),~(\asset{A}\neq \asset{B})\\
~& \Sigma_{\asset{B}\in\assetset}y_{\asset{A},\asset{B}} \geq (1-\varepsilon)\Sigma_{\asset{B}\in\assetset}y_{\asset{A},\asset{B}}~~~ \forall \asset{A}
\end{align}

The Stellar implementation charges no transaction commission (i.e.,
sets $\varepsilon$ to $0$) in its \SPEEDEX{} deployment.
This makes the linear program into an instance of the maximum
circulation problem (i.e.,
variable $y_{\asset{A},\asset{B}}$ denotes the flow from vertex $\asset{A}$ to vertex $\asset{B}$).
It is well known that the constraint matrices of these problems are totally unimodular (Chapter 19, Example 4 \cite{schrijver1998theory}).
This means that it always has an integral solution (Theorem 19.1, \cite{schrijver1998theory}) 
and can be solved by specialized algorithms 
(such as those outlined in \cite{kiraly2012efficient}).
Some of these algorithms run substantially faster than general simplex-based solvers.

%% file: decomposition.tex
\section{Market Structure Decomposition}
\label{sec:decomposition}

Suppose that the set of goods could be partitioned between a set of numeraires,
which might be traded with any other asset, and a set of stocks, which are only traded
with one of the pricing assets.

Then \SPEEDEX{} could compute a batch equilibrium by first computing an equilibrium taking into account 
only trades between pricing assets, then computing an equilibrium exchange rate for every stock between
the stock and its pricing asset, and finally combining the results.

More specifically:

\newcommand{\stockset}{\mathfrak{S}}

\begin{theorem}

Let $\assetset$ be the set of numeraires and $\stockset$ the set of stocks.  A stock $\asset{S}\in\stockset$ is traded with 
asset $a(\asset{S})\in \assetset$.

Suppose $(p,x)$ is an equilibrium for the restricted market instance considering only the numeraires.
For each $\asset{S}\in \stockset$, let $(r, y)$ be an equilibrium for the restricted market instance considering
only $\asset{S}$ and $a(\asset{S})$.

Then $(p^\prime, x^\prime)$ is an equilibrium for the entire market instance,
where 
\begin{enumerate}
\item $p^\prime_{\asset{A}}=p_{\asset{A}}$ for $\asset{A}\in \assetset$
\item $p^\prime_{\asset{S}}=\left(r_{\asset{S}}/r_{a(\asset{S})}\right)p_{a(\asset{S})}$
\item $x^\prime_{\asset{A},\asset{B}}=x_{\asset{A},\asset{B}}$ for $\asset{A},\asset{B}\in \assetset$
\item $x^\prime_{\asset{S}, a(\asset{S})}=y_{\asset{S}, a(\asset{S})}$
\item $x^\prime=0$ otherwise
\end{enumerate}

\end{theorem}

\begin{proof}

More generally, let $G$ be a graph whose vertices are the traded assets and which contains an edge $(\asset{A},\asset{B})$ if
$\asset{A}$ and $\asset{B}$ can be traded directly.

Decompose $G$ into an arbitrary set of edge-disjoint subgraphs $\lbrace G_i\rbrace$, such that any two subgraphs $G_i,G_j$
share at most one common vertex.  Then define a graph $H$ whose vertices are the subgraphs $G_i$, and where a subgraph $G_i$ is connected
to $G_j$ if $G_i$ and $G_j$ share a common vertex.  

If $H$ is acyclic, then an equilibrium can be reconstructed from equilibria computed independently on each $G_i$.

We reconstruct a unified set of prices iteratively, traversing along $H$.  Given adjacent $G_i$ and $G_j$ sharing common vertex $v_{ij}$,
 let $(p^i, x^i)$ and $(p^j, x^j)$ be
equilibria on $G_i$ and $G_j$, respectively, rescale all of the prices $p^j$ by $p^i_{v_{ij}}/p^j_{v_ij}$.

This rescaling constructs a new equilibria ($p^{j\prime}, x^j$)for $G_j$ that agrees with that of $G_i$ on the price of the shared good. 
As such, the combined system ($p^i\cup p^{j\prime}, x^i\cup x^j$) forms an equilibrium for $G_i\cup G_j$.

This iteration is possible precisely because $H$ is acyclic (a cycle could prevent us from finding a rescaling of some subgraph
that satisfied two constraints on the prices of shared vertices).

\end{proof}

%% file: alternate_strategies.tex
\section{Alternative Batch Solving Strategies}
\label{sec:alternatestrats}

\subsection{Convex Optimization}

We implemented the convex program of Devanur et al. \cite{devanur2016rational} directly,
using the CVXPY toolkit \cite{diamond2016cvxpy} backed by the ECOS convex solver \cite{domahidi2013ecos}.
Figure \ref{fig:cvxpygraph} plots the runtimes we observed to solve the problem while varying the number of assets and offers.

The runtimes are not directly comparable to those of \Tat{}---namely, this strategy does not
have the potential to shortcircuit operation upon early arrival at an equilibrium (our notions of approximation error also do not directly
translate to the notions used interally in the solver), nor is it optimized
for our particular class of problems. 

The important observation is that the runtime of this strategy scales linearly in the number of trade
offers.  Instances trading 1000 offers, for example, take roughtly 10x as long as instances trading only 100 offers.

\begin{figure}
\centering
\includegraphics[width=\columnwidth]{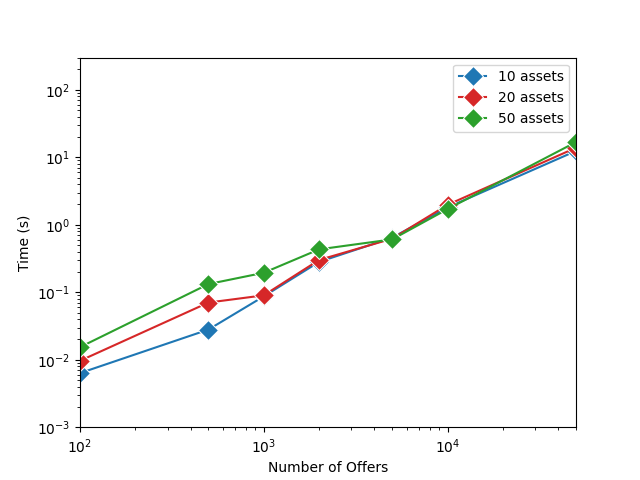}

\caption{
    Time to solve the convex program of Devanur et al. \cite{devanur2016rational} using the CVXPY toolkit \cite{diamond2016cvxpy},
    varying the number of assets and offers.
    \label{fig:cvxpygraph}
}

\end{figure}

This is not a surprising result, given that the number of variables in the convex program
scales linearly with the number of trade offers.

The choice of solver strategy does not, of course, change the structure of the input problem instances.
The same observation used in \S \ref{sec:logtransform} makes it possible to refactor the convex program so that the number of variables
does not depend on the number of open offers, and so that the objective (and its derivatives) can be evaluated in time
logarithmic in the number of open offers.

Unfortunately, this transformation makes the objective nondifferentiable.  The demand smoothing tactic of \S \ref{apx:demandsmoothing}
gives a differentiable but not twice differentiable objective (and presents challenges regarding numerical stability of the derivative).
Construction of a convex objective that approximates that of \cite{devanur2016rational}
while maintaining sufficient smoothness and numerical stability is an interesting open problem.

\subsection{Mixed Integer Programming}
\label{sec:dfusionmip}

Gnosis (Walther, \cite{dfusion}) give several formulations of a batch trading system as
mixed-integer programming problems.  These formulations track token amounts as integers
(instead of as real numbers, as used in \Tat{}'s underlying mathematical formulation),
which enables strict conservation of asset amounts with no rounding error.

However, mixed-integer problems appear to be computationally difficult to solve.
Walther \cite{dfusion} finds that the runtime of this approach scales faster than linearly.
Instances with more than a few hundred assets appear to be intractable for practical systems.

%% file: tatonnement_preprocessing.tex
\section{\Tat{} Preprocessing}
\label{sec:tatpreprocessing}

We include this section so that this paper can provide a comprehensive
 reference for anyone to develop their own \Tat{} implementation.

Every demand query in \Tat{} requires computing, for every asset pair,
the amount of the asset available for sale below the queried exchange rate.
As discussed in \S \ref{sec:tatimpldetails}, \Tat{} lays out contiguously in memory
all the information it needs to return this result quickly.

For a version of \Tat{} without the demand smoothing of \S \ref{apx:demandsmoothing},
a demand query for exchange rate $p$ (i.e. the ratio of the price of the sold asset to the price of the purchased asset)
\begin{equation}
\label{eqn:nonsmooth}
\Sigma_{i:mp_i \leq p} E_i
\end{equation}
where $mp_i$ denotes the minimum price of an offer
$i$ and $E_i$ denotes the amount of the asset offered for sale.

We can efficiently answer these queries by computing expression \ref{eqn:nonsmooth}
for every price $p$ used as a limit price

Demand smoothing complicates the picture.  The result of a demand query (with smoothing parameter $\mu$)
\begin{equation}
\label{eqn:smooth}
\Sigma_{i:mp_i < p(1-\mu)} E_i + \Sigma_{i:p(1-\mu) \leq mp_i \leq p} E_i * (p-mp_i)/(p\mu)
\end{equation}

We can rearrange the second term of the summation into
\begin{equation}
1/(p\mu) \Sigma_{i:p(1-\mu) \leq mp_i \leq p} (pE_i - E_imp_i)
\end{equation}

With this, we can efficiently compute the demand query
after precomputing, for every unique price $p$ that is used as a limit price,
both expression \ref{eqn:nonsmooth} and
\begin{equation}
\Sigma_{i:mp_i < p} mp_iE_i
\end{equation}

The division in equation \ref{eqn:smooth} can be avoided by recognizing that
\Tat{} normalizes all asset amounts by asset valuations (so equation \ref{eqn:smooth} is always
multiplied by $p$).

%% file: ppad.tex
\section{Buy Offers are PPAD-hard}
\label{sec:ppad}

A natural type of trade offer is one that offers to 
sell any number of units of one good to buy a fixed amount of a good (subject
to some minimum price constraint).  We call these \emph{limit buy
offers}.

\begin{example}[Limit Buy Offer]
\label{ex:buyoffer}

A user offers to buy 100 USD in exchange for EUR, selling as few EUR as possible
and only if one EUR trades for at least $1.1$ USD.
\end{example}

These offers unfortunately do not satsify a property known as ``Weak Gross Substitutability'' (WGS,
see
e.g., \cite{codenotti2005market}).  This property captures the core logic of \Tat{}.  If the price of one good
rises, the net demand for that good should fall, and the net demand for every other good should rise (or at least, not
decrease).  Limit sell offers satisfy this property, but limit buy offers do not.

\begin{example}
The demand of the offer in of example \ref{ex:buyoffer}, when
$p_{EUR}=2$ and $p_{USD}=1$, is $(-50~\text{EUR}, 100~\text{USD}).$

If $p_{USD}$ rises to $1.6$, then the demand for the offer is
$(-80~\text{EUR}, 100~\text{USD})$.

Observe that the price of USD rose and the demand for EUR fell.
\end{example}

Informally speaking, if offers do not satisfy the core logic of \Tat{}'s price update rule,
then \Tat{} cannot handle them in a mathematically sound manner.

More formally, Chen et al.~\cite{chen2017complexity} show through Theorem~7 and Example 2.4 that
markets consisting of collections of limit buy offers are PPAD-hard.
These theorems are phrased
in the language of the Arrow-Debreu exchange market model; see \S \ref{sec:speedexandad} for the 
correspondence between \SPEEDEX{} and this model.  In fact, the utility functions
used in Example 2.4 to demonstrate an example ``non-monotone'' (i.e., defying WGS) instance
are of the type that would arise by mapping limit buy offers into the Arrow-Debreu exchange market model.

%% file: filtering_performance.tex
\section{Deterministic Filtering Performance}

\label{sec:filtering}

\XXX{emphasize that this happens in a single pass, both here and in limitations.tex}

The deterministic transaction batch pruning system works by eliminating 
the transactions from all of the accounts that could create an unresolvable
conflict.  To be specific, if the sum of the amount of an asset used (either sent in a payment option
or locked to create a offer) by all of an account's transactions exceeds that
account's balance, then that account's transactions are removed.
If an account sends two transactions with the same sequence number (both of which have valid signatures,
and the sequence numbers are higher than the sequence number of the account's most recent transaction),
or two transactions cancel the same offer ID,
then that account's transactions are removed.
If two transactions create the same account ID, then both transactions are removed.

We generated batches of 400,000 transactions from the same synthetic transaction model as in \S \ref{sec:scalable},
and then duplicated 100,000 transactions at random to create a batch of 500,000.
A small number of accounts (1000) send transactions with conflicting sequence numbers.
We initialize the database (again, 10 million accounts) to give each account a small amount of money,
and a small number (one or two hundred) of accounts attempt to overdraft.

This filtering takes 0.13s and 0.07s seconds with 24 and 48 threads, respectively (averaged over 50 trials, after a warmup),
 giving a 21.0$\times$ and 38.4$\times$ speedup over the serial benchmark.  On a more contested benchmark, with only 10,000 accounts (almost all of which overdraft)
 the maximum speedup over the single threaded trial is only 5.3$\times$, but the overall filtering runtime is still just 0.10s.  
 Our implementation of the filtering is not heavily optimized, but in either parameter setting, the overhead is small.

%% file: bstm_compare.tex
\section{Block-STM Baseline}
\label{sec:bstm_baseline}

To provide a baseline for the measurements in Fig. \ref{fig:blockstm}, we also ran Block-STM
on our hardware (with hyperthreading disabled, as in \cite{gelashvili2022block}).  Fig. \ref{fig:blockstm_perf} displays the results.

\begin{figure}
\centering
\includegraphics[width=\columnwidth]{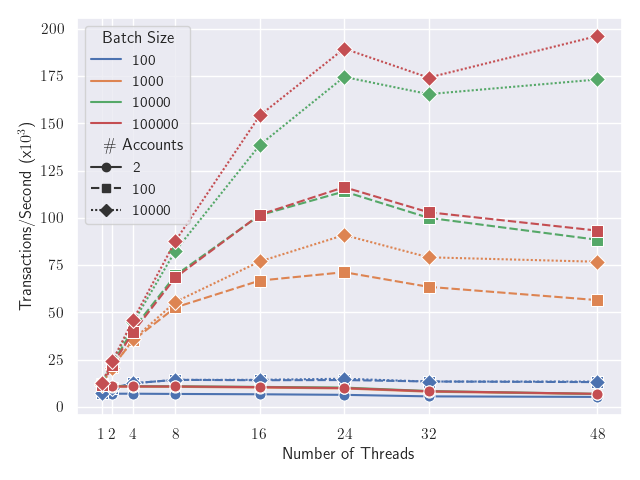}
\caption{
    Throughput of Block-STM on batches of ``Aptos p2p'' transactions with varying thread counts
 (average of 100 trials).
    \label{fig:blockstm_perf}
}
\end{figure}

These performance measurements are similar, quantitatively and qualitatively, to those reported in \cite{gelashvili2022block} (on different hardware).
Note that performance appears to reach a maximum after approximately
16 to 24 threads, and, unlike \SPEEDEX{}, does not effectively use additional hardware beyond this point,
even on relatively low-contention workloads.

%% file: extra_design.tex
\section{Additional Implementation Details}
\label{sec:additional_implementation}

\input{data_organization}

\input{data_storage}

\input{block_header}

\input{replay_prevention}

\input{fast_sorting}

\input{nd_block_assembly}

%% file: data_organization.tex
\subsection{Data Organization}
\label{sec:dataorg}

Account balances are stored in a Merkle-Patricia trie.  However, because a trie is not self-rebalancing, its worst-case adversarial lookup performance can be slow.  
As such, we store account balances in memory indexed by a red-black tree, with updates pushed to the trie once per block.

For each pair of assets ($\asset{A}$, $\asset{B}$), we build a trie
storing offers selling asset $\asset{A}$ in exchange for $\asset{B}$\@.  
Finally, in each block, we build a trie logging which accounts were modified.

We store information in hashable tries so that nodes can efficiently
compare their database state with another replica's
(to validate consensus and check for errors), and construct short proofs for users about exchange state.


%% file: data_storage.tex
\subsection{Data Storage and Persistence}
\label{sec:persistence}

\SPEEDEX{} uses a combination of an in-memory cache and
ACID-compliant databases (several LMDB\cite{lmdb} instances).  This choice suffices
for our experiments, but a database that persists data in epochs, like
Silo \cite{tu2013speedy}, or is otherwise optimized
for batch operation
might improve performance.

Our implementation uses one LMDB instance for the set of open offers, one instance for Hotstuff logs,
one instance for storing block headers, and 16 instances for storing account states.  LMDB is single-threaded, and we find that
the throughput of one thread generating database writes does not keep up with \SPEEDEX{}.  Accounts are randomly divided between these instances, 
according to a hash function keyed by a (persistent) secret key (which
is different per blockchain node).  This key must be kept secret so as to prevent
nodes from denial of service attacks.

Processing transactions in a nondeterministic order complicates recovery from a database snapshot where a block has been partially applied.
Cancellation transactions, in particular, refund to an account the remainder of an offer's asset amount.  We therefore cannot recover
if the snapshot of the orderbooks is more recent than the snapshot of the set of account balances, and our implementation takes care
to commit updates to the account LMDB instances before committing
updates to the orderbook LMDB\@.

%% file: block_header.tex
\subsection{Follower Optimizations}
\label{sec:header}

A block proposal includes the output of \Tat{} and the linear program in 
(the prices and trade amounts, as in \S \ref{sec:theorytopractice}).
This permits the nondeterminism in \Tat{} (\S \ref{sec:multiinstance}),
and lets the other nodes skip the work of running \Tat{}.

Proposals also include, for every pair of assets, the trie key
of the offer with the highest minimum price that trades in that block.  
When executing a proposal from another node,
a follower can compare the trie key of a newly created
offer with this marginal key and know immediately whether to make a trade
or add the offer to the resting orderbooks.  A node also
defers all checks that an account balance is not overdrafted to after it has 
executed all the transactions in a block.

%% file: replay_prevention.tex
\subsection{Replay Prevention}
\label{sec:replay}

Transactions have per-account sequence numbers to ensure a transaction
can execute only once.  Many blockchains require sequence numbers from
an account to increase strictly sequentially.  Our implementation
allows small gaps in sequence numbers, but restricts sequence numbers to
increase by at most an arbitrary limit (64) in a given block.  
Allowing gaps simplifies
some clients (such as our open-loop load generator), but more
importantly lets validators efficiently track consumed sequence
numbers out of order with a fixed-size bitmap and hardware atomics.

The Stellar implementation requires strictly consecutive sequence
numbers, mostly for backwards compatibility.

%% file: fast_sorting.tex
\subsection{Fast Offer Sorting}
\label{sec:fastsorting}
The running times of \S \ref{sec:tateval} do not include times to
sort or preprocess offers.  Na\"ively sorting large lists takes a long
time.  Therefore, we build one trie storing offers per asset pair, 
and we use an offer's price, written in big-endian, as
the first 6 bytes of the offer's 22-byte trie key.
Constructing the trie thus automatically sorts offers by price.

Additionally, \SPEEDEX{} executes offers with the lowest minimum
prices, so a set of offers executed in a round forms a dense (set of) subtrie(s), 
which is trivial to remove.

%% file: nd_block_assembly.tex
\subsection{Nondeterministic Block Assembly}
\label{sec:nd_assembly}

As discussed in \S \ref{sec:commutativesemantics}, \SPEEDEX{} must
assemble blocks of transactions in a manner that guarantees no account
is overdrafted after applying all of the transactions in the block.
The block proposal system (Fig. \ref{fig:diagram}, 2) manages this by
carefully controlling writes to shared state.

The proposal module takes as input a set of unconfirmed transactions
(the ``mempool'', in typical blockchain parlance) and outputs a
proposed block containing a subset of the unconfirmed transactions.
For each candidate unconfirmed transaction, a thread reserves the
ability to perform all necessary modifications by ``locking'' all
relevant data elements.  Once a transaction acquires all of its locks,
it performs its necessary state modifications and finally releases
the locks.
If it cannot acquire all necessary locks, it
releases any locks and excludes the transaction from the proposed
block.

Conceptually, a transaction offering a trade or sending a payment must
lock the number of units of assets that could be debited from the
account if the operation succeeds.  However, doing this with spinlocks
would preclude the scalability displayed in Figure~\ref{fig:blockstm}.
Instead, most reservations are performed with hardware atomics to
decrement the number of available units.  
Crediting an account can never fail because \SPEEDEX{} caps
the total amount of any asset issued at \texttt{INT64\_MAX}.
This process is conservative
in that it may reject transactions that could have executed
safely.


Unique offer IDs ensure that no offer is created twice, and atomic
boolean flags ensure an offer cannot be cancelled twice.  Sequence
numbers can be reserved by atomic bitmaps (as in \S \ref{sec:replay}).
For simplicity, our implementation does use exclusive locks when
creating new accounts (which we assume occurs relatively
infrequently).

%% file: morereplicas.tex
\section{Additional Replicas}
\label{sec:10rep}

\SPEEDEX{} invokes a consensus protocol no more than once per second in our experiments. 
To demonstrate that this overhead is negligible, we ran \SPEEDEX{} with $10$ replicas,
although with weaker hardware per replica, due to resource limitations.  Each replica is one
AWS c5ad.16xlarge instance, with one AMD EPYC 7R32 processor (48 CPUs @ 2.8Ghz per physical chip, 32 of which are allocated to our instances),
 128 GB of memory,
and two 1.1TB NVMe drives in a RAID0 configuration.  Performance measurements are plotted
in Figure \ref{fig:10rep}.

\begin{figure}  
\centering
\includegraphics[width=\columnwidth]{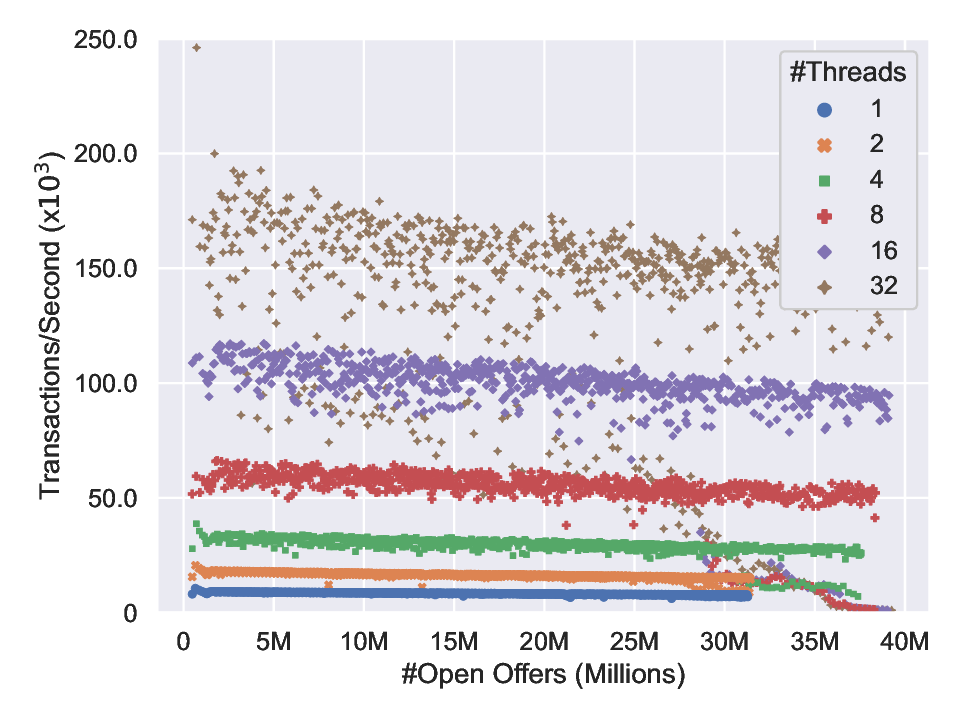}
\caption{
    Transactions per second on \SPEEDEX{} when running with 10 replicas (on weaker hardware than in Fig. \ref{fig:e2e}), 
    plotted over the number of open offers.
    \label{fig:10rep}
}
\end{figure}

The overall throughput numbers are lower here than in Figure \ref{fig:e2e} due to the weaker hardware,
but the scalability trends are the same.  Doubling the thread count increases performance by a factor
of between 1.8x and 1.9x, except that the jump from 16 to 32 gives a roughly 1.4x increase due to contention
with background tasks (particularly logging to persistent storage).

This graph also highlights how \SPEEDEX{} responds to insufficient hardware resources.  As the number of open offers
increases, \SPEEDEX{}'s memory requirements increase.  Eventually, memory starts to be paged to disk,
which dramatically increases disk usage and contends with the logging to persistent storage.
\SPEEDEX{} slows down in response, to ensure for safety that data in peristent storage is never too far out of sync.